\documentclass{amsart}

\usepackage{epsfig}
\usepackage{amsmath}
\usepackage{amsthm}
\usepackage{amsfonts}
\usepackage{array}
\usepackage{natbib}

\newtheorem{proposition}{Proposition}[section]
\newtheorem{corollary}{Corollary}[section]
\newtheorem{theorem}{Theorem}[section]

\theoremstyle{definition}
\newtheorem{observation}{Observation}[section]
\newtheorem{property}{Property}
\newtheorem{definition}{Definition}[section]
\newtheorem{case}{Case}

\newcommand{\ignore}[1]{}

\bibpunct{[}{]}{;}{a}{,}{,}

\title{Dynamic Maintenance of Half-Space Depth for Points and Contours}

\thanks{This research was partially supported by NSF grants \#CCF-0431027 and \#CCF-0830734}

\author{Michael A. Burr}
\address{Department of Computer Science, Tufts University, Medford, MA 02155, USA}
\curraddr{Mathematics Department, Fordham University, Bronx, NY 10458.}
\email{burr@cims.nyu.edu}

\author{Eynat Rafalin}
\address{Department of Computer Science, Tufts University, Medford, MA 02155, USA}
\curraddr{Google Inc., 1600 Amphitheatre Parkway, Mountain View, CA 94043, USA}
\email{eynat.rafalin@alumni.tufts.edu}

\author{Diane L. Souvaine}
\address{Department of Computer Science, Tufts University, Medford, MA 02155, USA}
\email{dls@cs.tufts.edu}

\begin{document}

\begin{abstract}
Half-space depth (also called Tukey depth or location depth) is one of the most commonly studied data depth measures because it possesses many desirable properties for data depth functions.  The data depth contours bound regions of increasing depth.  For the sample case, there are two competing definitions of contours: the rank-based contours and the cover-based contours.

In this paper, we present three dynamic algorithms for maintaining the half-space depth of points and contours: The first maintains the half-space depth of a single point in a data set in $O(\log n)$ time per update (insertion/deletion) and overall linear space.  By maintaining such a data structure for each data point, we present an algorithm for dynamically maintaining the rank-based contours in $O(n\cdot\log n)$ time per update and overall quadratic space. The third dynamic algorithm maintains the cover-based contours in $O(n\cdot \log^2 n)$ time per update and overall quadratic space.

We also augment our first algorithm to maintain the local cover-based contours at data points while maintaining the same complexities.  A corollary of this discussion is a strong structural result of independent interest describing the behavior of dynamic cover-based contours near data points.

\end{abstract}

\maketitle
\section{Introduction}
\label{sec:intro}
{\em Data depth} is a statistical analysis method that assigns a {\em depth value} to a point $x\in\mathbb{R}^d$ that measures how deep (or central) $x$ is relative to a probability distribution or to a data set (see surveys \citep{Aloupis:DIMACS:DataDepth,Liu:Survey}).  Data depth has recently attracted attention from the computational geometry community \citep{Book:DIMACS:DataDepth}.  It does not require {\em a priori} assumptions on the underlying probability distribution of data and handles observations that deviate from the data (outliers).  {\em Half-space depth}, first defined by \citep{Tukey:half-space,Hodges:Bivariate}, is a data depth measure that has been commonly studied in recent years, e.g., \citep{Liu:Survey,Bremner:half-spaceDepth,Aloupis:LowerBounds,Miller:Contours}.  {\em Depth contours} are nested contours that enclose regions of increasing depth \citep{Tukey:half-space}.  They provide a powerful tool to visualize, quantify, and compare data sets (e.g., \citep{Liu:Survey,Rousseeuw:Bagplot}).  The statistics community has produced two competing definitions for depth contours, yielding distinctive contours.  The two main approaches were termed {\em cover} and {\em rank} \citep{Rafalin:Contours}, see Figure~\ref{fig:contours}.  In this paper, we present dynamic algorithms in $\mathbb{R}^2$ for maintaining the half-space depth of a data point in $O(\log n)$ time per update, the rank-based contours in $O(n\log n)$ time per update, and the cover-based contours in $O(n\log^2n)$ time per update.

The half-space depth (sometimes called the location depth or the Tukey depth) of $x\in\mathbb{R}^d$ relative to a data set $\mathcal{F}_n=\{X_1,\cdots,X_n\}$ is the minimum fraction of points of $\mathcal{F}_n$ lying in any closed half-space containing $x$.  In $\mathbb{R}^2$, the half-space depth of a {\em single point} can be computed in optimal $O(n\log n)$ time \citep{Aloupis:LowerBounds,Ruts:Contours,Bremner:half-spaceDepth} and the half-space depth of all points can be computed in $O(n^2)$ time \citep{Miller:Contours}.  In higher dimensions, \citep{Fukuda:Parallel,Bremner:half-spaceDepth,Rousseeuw:HighDimensions,Struyf:HighDimensions}\nocite{Book:Optimization} present algorithms to compute the half-space depth of a single point.  The {\em half-space median} is a point of maximal half-space depth (if more than one point has maximal depth, then the center of mass of all such points is commonly used); in $\mathbb{R}^d$, the half-space depth of the median is at least $\left\lfloor\frac{n}{d+1}\right\rfloor/n$ by Helly's theorem \citep{Eckhoff:Helly}\nocite{Book:Handbook:Convex} and at most $\lceil n/2\rceil/n$ (for data sets without double points).  In $\mathbb{R}^2$, a point with maximal depth can be computed in $O(n\log^3 n)$ time \citep{Langerman:Median}, and \citep{Chan:Randomized} presents a randomized algorithm to compute the median in $O(n\log n)$ expected time.  Chan's algorithm extends to higher dimensions with an expected time of $O(n^{d-1})$.  For other results, see \citep{Ruts:Bivariate,Ruts:Bivariate2,Matousek:Median,Langerman:Maximal}.  A {\em center point} is a point with depth at least $\left\lfloor\frac{n}{d+1}\right\rfloor/n$.  In the plane, \citep{Jadhav:Centerpoint} present an $O(n)$ algorithm for computing a center point.  For other results, see \citep{Cole:Hulls,Matousek:Median,Cole:Centerpoint2,Agarwal:Center}

The {\em sample cover contour of depth $d/n$} \citep{Tukey:half-space} is the boundary of all points (not necessarily from the data set) with depth at least $d/n$.  The idea behind the cover-based contours is based on the assumption that the data set is generated from an underlying distribution and that the contours should represent the behavior of the underlying distribution.  The {\em sample rank contour of the $\alpha$th central region} is the convex hull containing the $\alpha$ most central fraction of data points \citep{Liu:Survey}. The rank-based contours focus more on the data points themselves because data points may represent actual observations, see Figure \ref{fig:contours}.  Therefore, only data points can be vertices of rank-based contours, and they may appear on multiple contours.  On the other hand, non-data points from $\mathbb{R}^d$ can be vertices of cover-based contours and each data point will occur on a single contour (for data sets in general position). For sufficiently nice underlying distributions, both types of contours will, almost surely, approach the half space depth contours for the underlying distribution \citep{Zuo:Structural}.  Although the rank-based contours lack some of the structure of the cover-based contours, they often provide a reasonable and less computationally expensive contour.  In $\mathbb{R}^2$, any half-space depth cover-based contour edge lies on a segment between two data points.  At most $\binom{n}{2}$ such segments exist, thus the complexity of the collection of half-space cover contours is $O(n^2)$.  This bound is tight, and it is achieved, for example, when all points of the data set lie on the convex hull of the data set.  Rank-based depth contours, by definition, enclose the $1st,\frac{n-1}{n}th,\dots,\frac{1}{n}th$ sample regions, have a total complexity of $O(n^2)$ (once data points of the same depth have been arbitrarily ordered), but only involve $\Theta(n)$ distinct edges since many contours share edges.  (The bound on the number of edges can be proved by amortizing the complexity over the data points since the difference between successive contours is induced by removing a point and its two incident edges.)  In $\mathbb{R}^2$, \citep{Miller:Contours} present an optimal algorithm that computes all of the cover-based contours in $O(n^2)$ time; at the same time, it computes the depth of every data point, and, therefore, can also be used to compute rank-based contours.  Several implementations exist for computing cover-based contours: \textsc{isodepth} \citep{Ruts:Contours} can compute the $k^{\text{th}}$ contour in $O(n^2\log n)$ time.  \textsc{fdc} \citep{Johnson:Fast} computes the outermost $k$ contours and outperforms \textsc{isodepth} for small $k$.  Using hardware assisted computation, \citep{Krishnan:Hardware} compute the half-space depth cover-based contours for display.  \citep{Chan:Randomized} presents an algorithm to compute the cover-based depth contour of depth $k$ in $O(n\log^2 n)$ expected time.  In higher dimensions, \citep{Fukuda:Parallel} present an algorithm for the computation of these contours.

\begin{figure}[hbt]
\begin{center}
$\begin{array}{c@{\hspace{0.6in}}c}
\epsfysize=1.7in
\epsffile{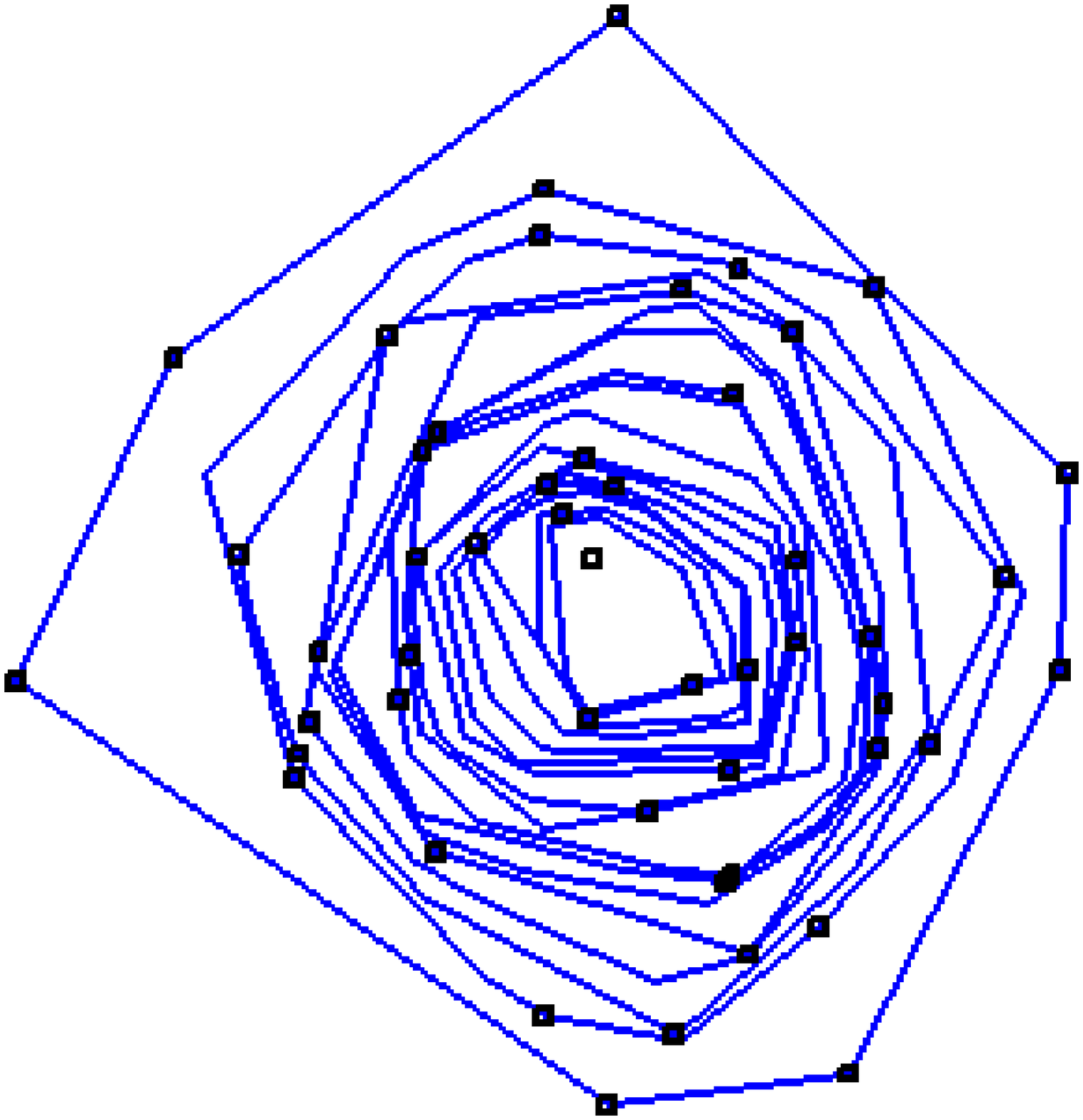} &
    \epsfysize=1.7in
    \epsffile{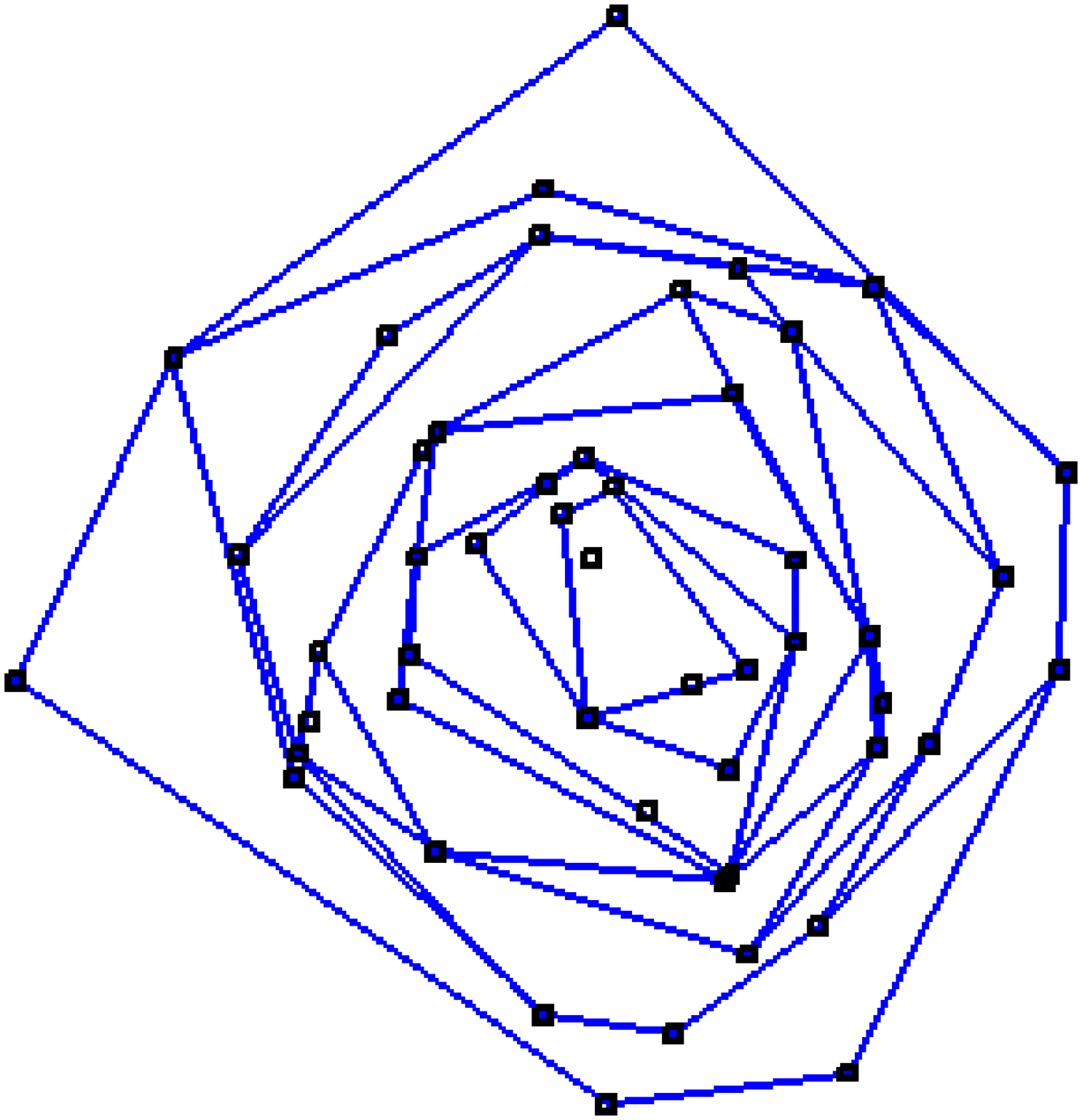} \\
\mbox{(a)} & \mbox{(b)}
\end{array}$
\end{center}
\caption{Half-space depth:
(a) All {\em cover-based} contours for 50 points
from a bivariate normal distribution with mean $(0,0)$ and covariance matrix $4 I$; (b) The $10\%, 20\% \dots 100\%$ {\em rank-based} contours for the same data set.  Both sets of contours have similar shapes, but are composed of distinct polygons: for example, the cover-based contours do not share vertices while the rank-based contours do share vertices.}
\label{fig:contours}
\end{figure}

Few dynamic algorithms (allowing the insertion or deletion of data points) for data depth computations exist, despite the interest of the statistical community \citep{Liu:Personal}.  Recently, \citep{DynamicHam} presented a dynamic algorithm which can be used to maintain a point of half-space depth at least $1/4$ and also a dynamic algorithm which deals with the case where the convex hull peeling depth \citep{Barnett:Peeling,Eddy:Peeling} of a data set is bounded.  Another result, presented in \citep{Agarwal:Dynamic}, dynamically maintains a data structure for half-space queries; this solves a different, but related, problem than to maintain the half-space depth of a point.  For cover contours, our algorithm is essentially maintaining the $k$-belts of an arrangement, as introduced in \citep{EdelsbrunnerWelzelBelts}.  In \citep{EdelsbrunnerWelzelBelts}, the authors construct the levels of an arrangement using a sweepline and claim, in passing, that the algorithm can be made dynamic.  While their algorithm is related to ours since we are solving the same problems and use similar techniques and tricks, it is difficult to compare the algorithms: the authors neither provide an explicit version of their dynamic algorithm nor any claims as to its complexity.

This paper presents the first dynamic algorithm for computing the half-space depth contours in $\mathbb{R}^2$ of a data set of $n$ points in general position.  In Section 2, we review some properties of half-space depth and half-space depth contours.  In Section 3, we present an algorithm for maintaining the half-space depth of a single point in $O(\log n)$ time per operation (insertion or deletion) and linear overall space.  The data structure that we use was originally presented in \citep{Kreveld:Regression}, but we augment it to enable dynamic updates.  In Section 4, we expand the algorithm and augment the data structure in Section 3 not only to maintain the depth of a point, but also to maintain the cover-based contours near points that do not lie on degenerate contours, i.e., contours consisting of either a single point or two points and the segment between them.  Through this analysis, we present a theorem of independent interest that shows that the cover-based contours cannot change too much after an insertion or deletion.  In Section 5, we use the algorithm from Section 3 (or from Section 4, since that algorithm is an augmented version of the algorithm from Section 3) to develop an algorithm to maintain the rank-based half-space depth contours in $O(n\log n)$ time per update and overall quadratic space.  In Section 6, we present an algorithm and data structure to maintain the cover-based half-space depth contours in $O(n\log^2 n)$ time per operation and overall quadratic space.  Both of these algorithms represent improvements over the static running time of $O(n^2)$.  Finally, we conclude in Section 7.

\section{Half-Space Data Depth and Contours}
In this section, we review the basic definitions and properties of half-space depth and its contours.  In this direction, we provide alternate combinatorial descriptions of these definitions which will be useful throughout this paper.  For example, the global description of the contours will be used in Section \ref{sec:DynamicCover}, and the discussion of the description of the local half-space depth contours near a data point will be completed in Section \ref{sec:localcontours}.  All of the half-spaces considered in this paper are closed.

\begin{definition}[\citep{Tukey:half-space}]
The half-space depth of a point $p$ with respect to a data set $\mathcal{F}_n=\{X_1,\cdots,X_n\}$ is the smallest fraction of data points contained in a half-space containing $p$, i.e., $\min\{(\#H\cap\mathcal{F}_n)/n|p\in H\text{ a half-space}\}$.  It is sufficient to consider only half-spaces that include $p$ in their boundary, i.e.,  $\min\{(\#H\cap\mathcal{F}_n)/n|p\in\partial H,\text{ }H\text{ a half-space}\}$.
\end{definition}

There is some tension between how half-space depth is computed in the statistics and probability community and in the computational geometry community.  The statistics and probability viewpoint defines half-space depth in terms of probabilities over half-spaces \citep{Liu:Survey}.  The computational geometry viewpoint prefers to count the number of data points in a half-space.  In our definition, we choose to represent the half-space depth as the fraction, instead of the number, of points of $\mathcal{F}_n$ in the half-space $H$ because this definition more accurately reflects the relationship between the finite sample case and the continuous case, i.e., the definition for distributions \citep{Liu:Survey}.  Even though we use this definition in our theoretical proofs, in our algorithms, the standard computational geometry definition becomes more desirable to maintain our time complexities.

Half-space depth has many nice properties.  The half-space median is robust in the presence of outliers because Helly's theorem implies that at least $\left\lfloor\frac{n}{d+1}\right\rfloor$ of the data points in $\mathcal{F}_n$ must be outliers to alter the position of the median significantly.  Half-space depth has all of the desirable properties for a data depth function: affine invariance, maximality at the center, monotonicity relative to the deepest point, vanishing at infinity, and invariance under dimensions change \citep{Zuo:Structural,SimplicialDepth}.

\subsection{Half-Space Depth Contours}
There are two main definitions for data depth contours called the cover-based contours and the rank-based contours \citep{Rafalin:Contours}, and, for most data sets, they provide distinct contours.

\begin{definition}[Cover-based contours \citep{Tukey:half-space}]
The {\em cover-based contour of depth $t$} is the boundary of $\{x\in\mathbb{R}^d\text{ }|\text{ Half-space depth of $x$ is at least }t\}$.  This approach constructs a contour around all points of $\mathbb{R}^d$ (not necessarily from the data set) with half-space depth at least $t$.
\end{definition}

A simple, albeit inefficient, way to compute the cover-based contour of depth $t$ in $\mathbb{R}^d$ (which will be justified below) is as follows: first, divide the plane into cells determined by hyper-planes through $d$ data points.  The half-space depth is constant in each of these cells.  Second, pick a point in each cell (which cannot be a data point), and compute its half-space depth.  Finally, take the closure of the union of all cells with half-space depth at least $t$; the boundary of this region is the cover-based contour of depth $t$.

\begin{definition}[Rank-based contours \citep{Liu:Survey}]
The {\em rank-based contour of the $\alpha$th central region} is the convex hull of the $\alpha$ most central fraction of data points where ties are broken arbitrarily.  This approach constructs a contour around the deepest data points.
\end{definition}

A simple way to compute the rank-based contour of the $\alpha$th central region is to first compute the depth of every data point and sort the data points by depth (in this sorted list, ties are broken arbitrarily).  Then, take the $n\cdot\alpha$ data points with largest half-space depth (with respect to the arbitrary ordering) and compute the convex hull of this set.  This convex hull is the rank-based contour of the $\alpha$th central region.

Both contour definitions can be adapted to any data depth function by the appropriate change in the definition.  In this generality, both contour definitions have their merits, for example, the rank-based contours are always convex and nested while the cover-based contours have more structure, see \citep{SimplicialDepth}.  For half-space depth, in particular, the cover-based contours can be described in terms of a certain finite collection of half-spaces  determined by the data points.  On the other hand, the rank-based contours have no such description.  In the remainder of this section, we recall this description.  We begin with the following assumption: for the remainder of the paper, we assume that {\em all data sets are in general position}.

The definition of half-space depth is a minimization problem over all half-spaces that contain a given point $p$.  As is usual in this setup, there is a dual problem that describes all points of a given half-space depth, as in the following proposition:

\begin{proposition}\label{prop:inter}
The cover-based contour of depth $k/n$ in $\mathbb{R}^d$ is the boundary of the intersection of all closed half-spaces containing exactly $n-k+1$ data points.
\end{proposition}

This result can, in fact, be improved to only consider a finite number of half-spaces.

\begin{corollary}\label{cor:inter}
The cover-based contour of depth $k/n$ in $\mathbb{R}^d$ is the boundary of the intersection of all closed half-spaces containing exactly $n-k+1$ data points and whose boundary contains exactly $d$ data points.
\end{corollary}

These results are not difficult, but they both require some technical, linear algebra-type lemmas.  For example, one way of proving these equivalences is to use translations and rotations, see \citep{Cole:Hulls} for a sketch.  To translate a given half-space means to find a second half-space whose bounding hyperplane is parallel to the bounding hyperplane of the original half-space and both half-spaces open in the same direction.  To rotate a given half-space around a subflat of the bounding hyperplane means to find a second half-space where the bounding hyperplane includes the given subflat (alternately, this rotation can be interpreted as a point in the appropriate Grassmannian).  Note that Corollary \ref{cor:inter} does not say that the $d$ data points lie on the contour of depth $k/n$, only that the faces of the contour of depth $k/n$ lie in hyper-planes that each contain $d$ points, but these points may be far from the contour itself, e.g., see Figures \ref{fig:contours}(a) and \ref{fig:degenerate}(a).

In addition, Corollary \ref{cor:inter} provides an explicit way to compute the half-space depth of a data point by investigating a finite number of half-spaces.  In particular, we will use the following corollary (in 2-dimensions) to motivate our data structures:

\begin{corollary}\label{cor:definition}
If $p$ is a data point of $\mathcal{F}_n$ in $\mathbb{R}^d$, then to compute the half-space depth of $p$ it is sufficient to consider half-planes whose boundaries include both $p$ and $d-1$ other data points of $\mathcal{F}_n$, i.e.,  $\min\{(\#H\cap\mathcal{F}_n-(d-1))/n|\#\mathcal{F}_n\cap\partial H=d\text{ and }H\text{ a half-space}\}$.
\end{corollary}

Note that the $-(d-1)/n$ appears in the formula because a small rotation can remove the $(d-1)$ data points, other than $p$, from $\partial H$.  Since every data point appears on some contour, Corollary \ref{cor:definition} implies the following corollary which is an initial description for the local cover-based contours near a data point:

\begin{corollary}\label{cor:localcontours}
If $p$ is a data point of half-space depth $k/n$ in $\mathbb{R}^d$, then the edges of the cover-based contour of depth $k/n$ that are incident to $p$ are defined by planes through $p$ and $d-1$ other data points.
\end{corollary}

\subsection{Terminology}\label{sec:terminology}

Our discussion above has highlighted the data which we need in order to maintain the rank-based or cover-based contours.  In this section, we define the terminology for the half-spaces and their associated geometric objects which will be used through the remainder of this paper.  We begin with the following assumption: for the remainder of this paper, we {\em restrict our attention to $\mathbb{R}^2$}.

For most cover-based contours, the contour bounds a region, but, for the deepest contour, it is possible that the contour does not bound a region, but consists of either a single point or a segment.  Our algorithm must treat these degenerate contours as a special case.

\begin{definition}
A cover-based contour is {\em degenerate} if it consists of either a single point or of two points connected by a segment, see Figure \ref{fig:degenerate}(a).
\end{definition}

\begin{figure}[hbt]
\epsfysize=1.7in
\centerline{\begin{tabular}{c@{\hspace{.6in}}c}
\epsfbox{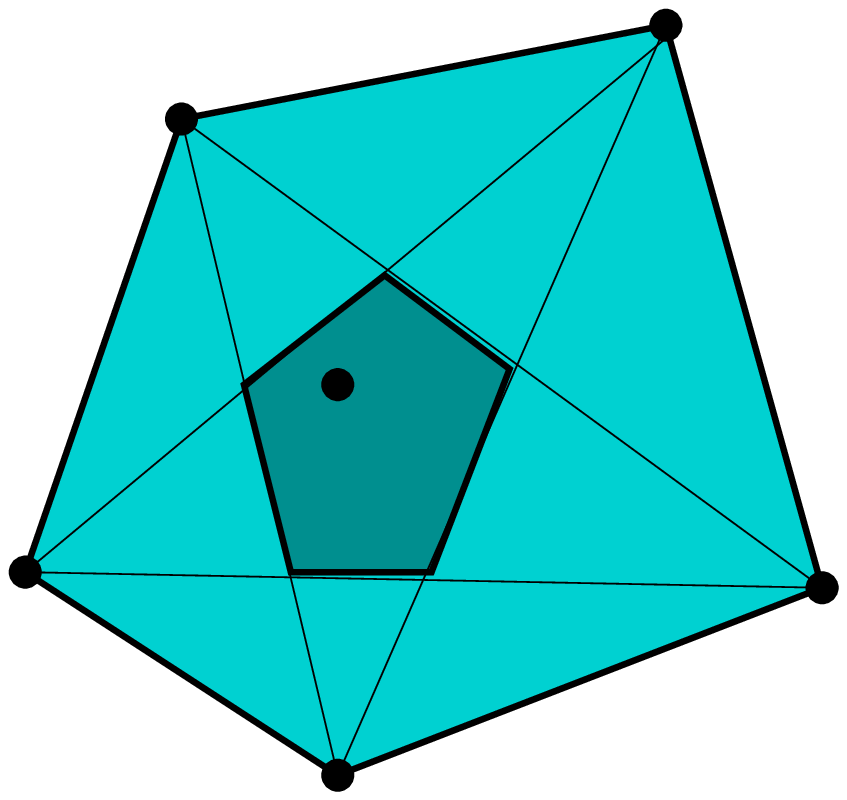}&\epsffile{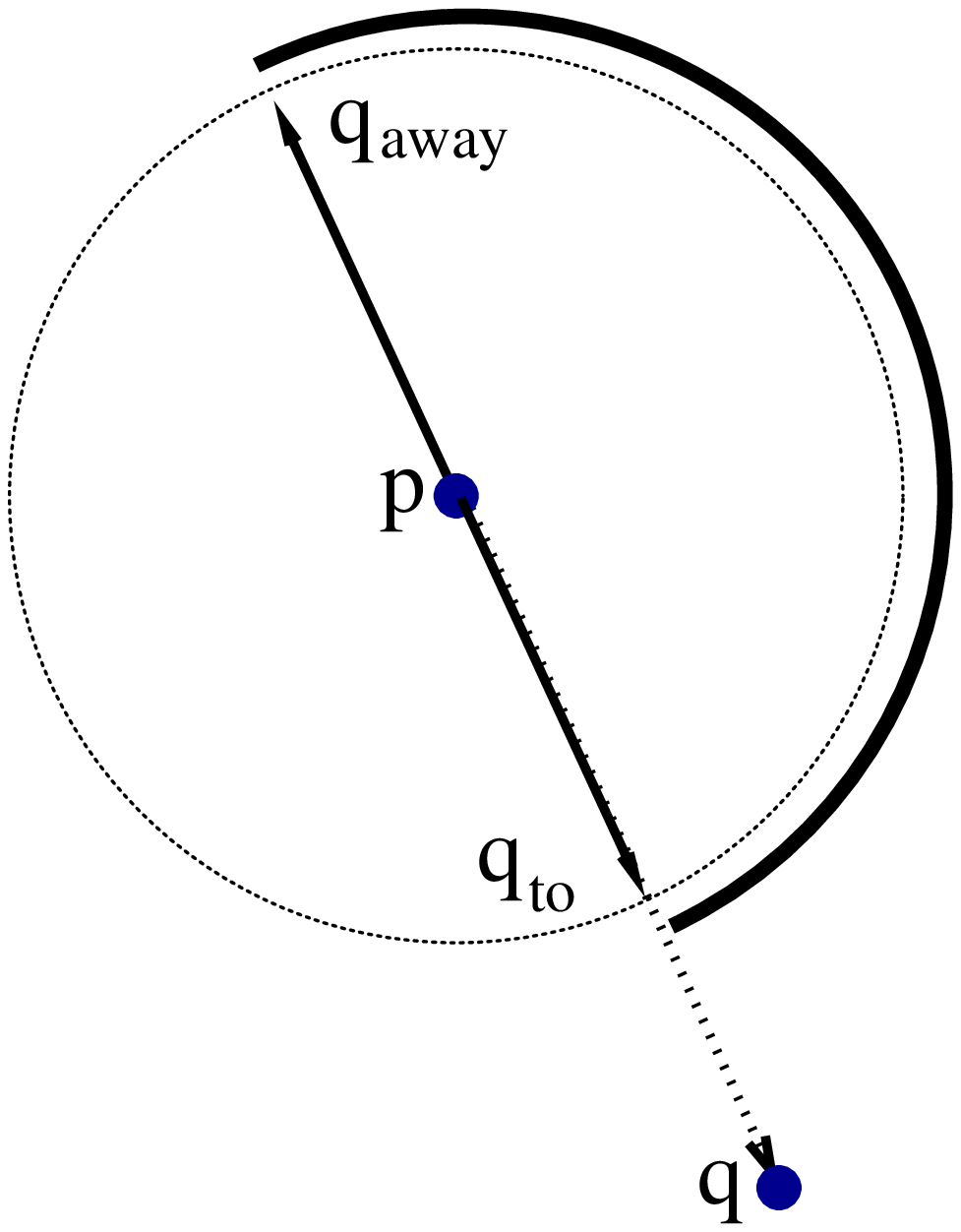}\\
(a)&(b)
\end{tabular}}
\caption{(a) Degenerate half-space depth contours: A set of 6 points. The outermost contour is the convex hull of the data set; the next contour contains 5 vertices, none a point of the data set. The inner-most contour contains one data point. (b) The transformation of the half-planes determined by $q$ and $p$.  The bold semicircle represents all half-planes whose number of points will be incremented (decremented) by one when $q$ is inserted into (deleted from) the data set because $q$ is to the right of the associated vector.}\label{fig:vectors}
\label{fig:degenerate}
\end{figure}

It is easy to show using Corollary \ref{cor:inter} that the only time that a degenerate cover-based contour consists of a segment is when the data set has exactly two points.  In non-degenerate cases, the half-space depth contours bound a region and we now define the features of these regions locally near a data point:

\begin{definition}
The {\em defining edges (lines)} of data point $p$ of half-space depth $k/n$ with respect to the data set $\mathcal{F}_n$, whose cover-based contour of depth $k/n$ is non-degenerate, are the two edges (lines defined by the edges) incident to $p$ on the cover-based contour of depth $k/n$.  We will use the notation $l_1$ and $l_2$ for the defining lines.
\end{definition}

Each of the defining lines for a point $p$ determines two half-planes, one of which contains $n-k+1$ data points while the other contains $k+1$ data points.  The half-plane containing $k+1$ data points is important because it describes the contour of depth $k/n$ locally and provides a witness for the depth of $p$, see Figure \ref{fig:contours2}.

\begin{definition}
The {\em defining half-planes} $H_{l_1}, H_{l_2}$ for data point $p$ of depth $k/n$ with respect to $\mathcal{F}_n$, whose cover-based of depth $k/n$ contour is non-degenerate, are the closed half-planes bounded by the defining lines of point $p$ with respect to $\mathcal{F}_n$ and which contain $k+1$ data points.
\end{definition}

Corollary \ref{cor:inter} implies that the defining lines for a data point $p$ must pass through an additional data point.  For our dynamic algorithms, it makes sense to keep track of all possible half-planes which could be new defining half-planes after an insertion or deletion, i.e., all half-planes whose boundary includes $p$ and an additional data point.

\begin{definition}
The {\em meaningful half-planes} are the half-planes whose bounding line passes through two data points.  Let $p$ be a data point, then the meaningful half-planes with respect to $p$ are all of the meaningful half-planes whose boundary includes $p$.  Note that there are $2(n-1)$ meaningful half-planes with respect to $p$ since every line determines two half-planes.
\end{definition}

Since is is not enough to use lines to keep track of meaningful half-planes, we choose to represent half-planes as (unit) vectors or, equivalently, the points on the unit circle centered at $p$ where the unit vectors terminate.  The half-plane {\em associated with a vector $v$} centered at $p$ is the half-plane consisting of all data points to the {\em right of the vector}, i.e., all points $q$ such that the clockwise angle between $v$ and the vector $v_q$, the vector from $p$ to $q$, is less than or equal to $\pi$, see Figures \ref{fig:vectors}(b) and \ref{fig:trans}

\begin{figure}[hbt]
\begin{center}
$\begin{array}{c@{\hspace{0.6in}}c}
\epsfysize=1.7in
\epsffile{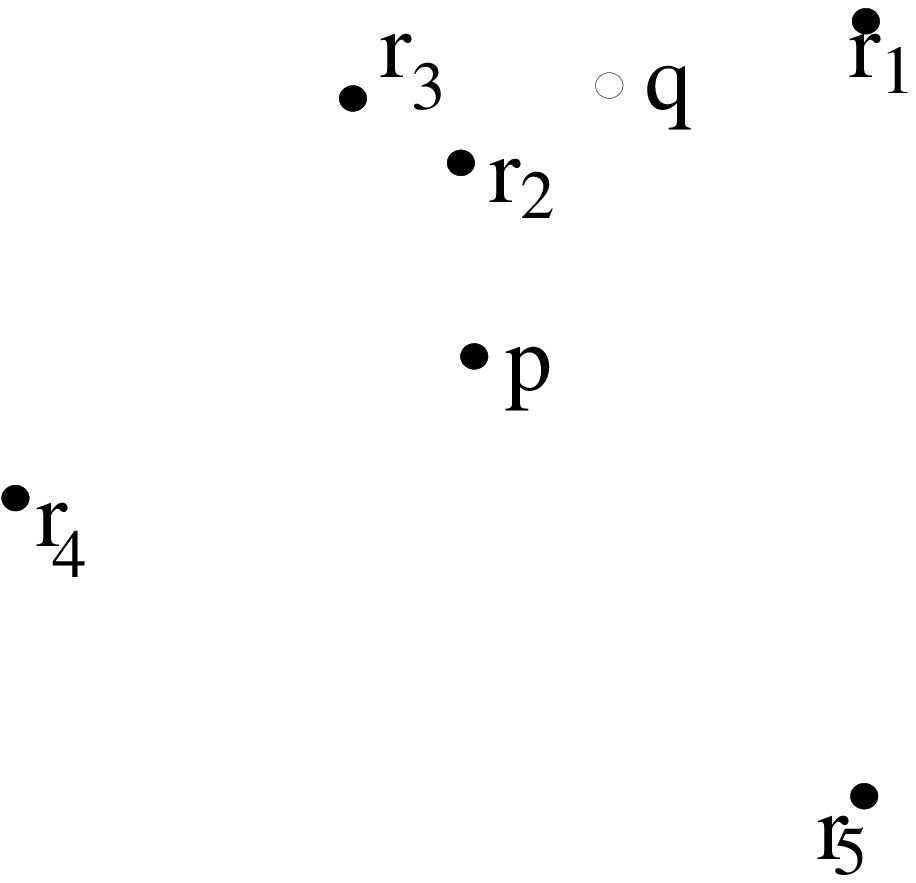} &
    \epsfysize=1.7in
    \epsffile{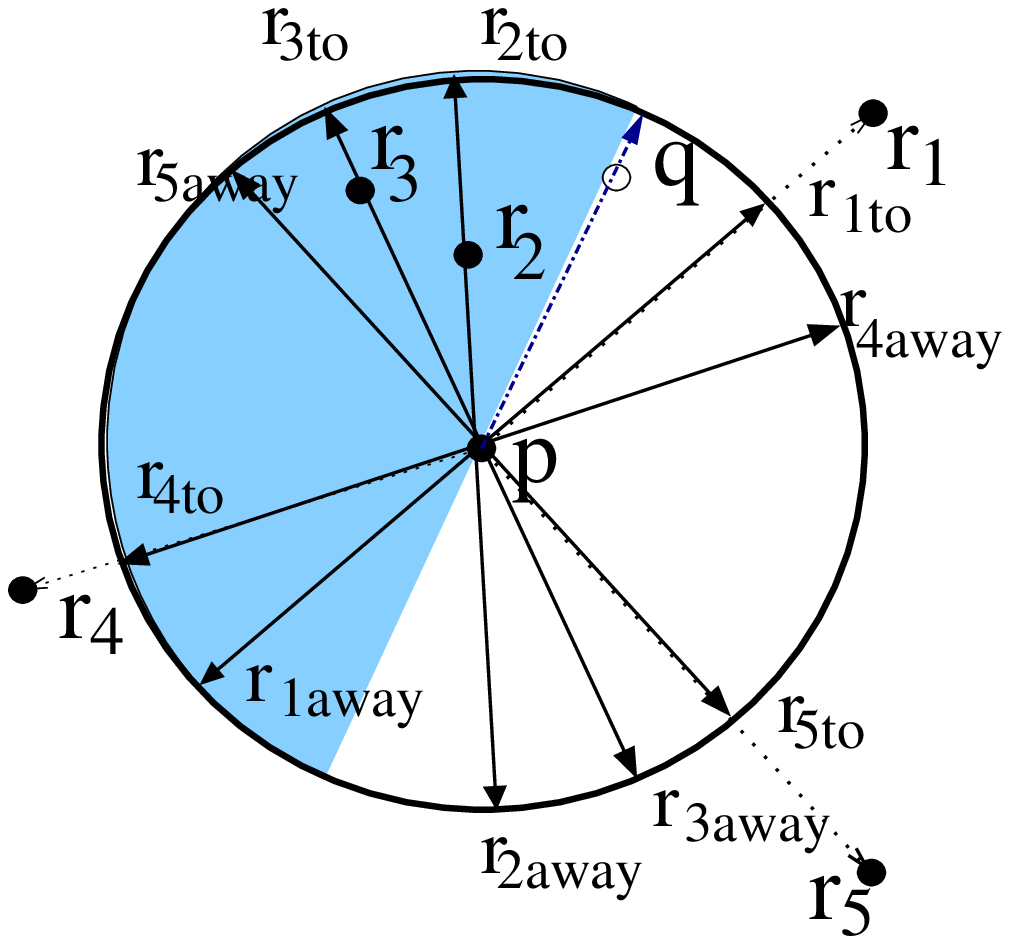} \\
\mbox{(a)} & \mbox{(b)}
\end{array}$
\end{center}
\caption{The transformation of the meaningful half-planes with respect to $p$: (a) Point $q$ is being inserted into (or deleted from) the data set.  (b) The vectors represent the two half-planes defined by each data point and $p$.  The half-plane associated with a vector consists of all points to its right so that the shaded region represents the set of half-planes which contain $q$.  In particular, the shaded region represents the set of half-planes whose number of included points increases (decreases) when $q$ is inserted into (deleted from) the data set.  On the other hand, $q$, $r_1$, and $r_5$ are in the half-plane associated with the vector pointing towards $q$.}\label{fig:trans}
\end{figure}

\section{Dynamic Maintenance the Depth of Data Points}\label{sec:dynamicdepthpoints}
In this section, we present an algorithm to maintain dynamically the half-space depth of a data point in $O(\log n)$ time per insertion or deletion and linear space overall.  This is an important step in the maintenance of the rank-based contours because these contours are based on the depths of the data points.  One of the simplest $O(n\log n)$ algorithms for computing the half-space depth of a point $p$, not necessarily a data point, is the following algorithm \citep{Aloupis:LowerBounds}:

First, sort all of the data points radially around $p$.  Next, use two pointers to step through this list.  The two pointers are separated by an angle of $\pi$ to represent a half-plane passing through $p$.  Finally, keep track of the smallest fraction of data points between the two pointers; this value, minus $1/n$, is the half-space depth of $p$ (the $-1/n$ appears as in Corollary \ref{cor:definition}).

In this section, we turn this static algorithm into a dynamic algorithm.  The static version of the data structure that we use was originally presented by \citep{Kreveld:Regression} for a different problem.

\subsection{The Angular Transformation}
The first step in the static algorithm above is to sort the data points radially around data point $p$.  In the dynamic algorithm, we instead need to maintain the sorted list of data points.  In the algorithm above, we needed two pointers to find the depth of $p$; therefore, our data structure will maintain all of the meaningful half-planes with respect to $p$.  We represent the half-planes as unit vectors originating at $p$, as described in Section \ref{sec:terminology}.  Since every data point, other than $p$, defines two meaningful half-planes, we label the vector pointing towards the point ``to'' and the vector pointing away from the point ``away,'' see Figures \ref{fig:vectors}(b) and \ref{fig:trans}.

The advantage of this representation is that the updates to the count of points in meaningful half-planes after an insertion or deletion are relatively simple to describe.  In fact, when a point $q$ is inserted into (deleted from) the data set, the only half-planes that need to be adjusted are those in the semicircle counter-clockwise from $q_{\text{to}}$ to $q_{\text{away}}$, see Figure \ref{fig:vectors}(b).  In the case of a circle, this semicircle is an arc and can be used to update the number of points in the half-planes in sub-linear time.

\subsection{Data Structure for Updating the Half-Space Depth of a Data Point}

In order to recompute the depth of point $p$ after an insertion or deletion, first, the number of data points in every meaningful half-plane with respect to $p$ has to be recomputed.  Then, the minimum number of data points in these meaningful half-planes must be computed.  If this number is $k+1$, then the depth of $p$ will be $k/(n\pm 1)$ where the $\pm$ corresponds to the insertion or deletion of a data point.  To recompute the depth of $p$ efficiently, all meaningful half-planes with respect to $p$ are updated simultaneously by increasing or decreasing the number of points associated with all affected half-planes by one.

We require a data structure that maintains all of the meaningful half-planes for a given point $p$ and which allows insertions and deletions in $O(\log n)$ time.  The static version of this data structure was originally presented in \citep{Kreveld:Regression} for a different problem.  We begin by describing the operations which are necessary for a data structure to solve this problem and then discuss our particular implementation (in Section 4, we will augment this data structure in order to maintain more of the local structure of the depth contours).

Let \texttt{T} be the data structure representing the set of meaningful half-planes with respect to $p$.  Each meaningful half-plane is labeled by a non-negative integer \texttt{Depth}, which is the number of data points in the half-plane.

\begin{itemize}
\item  \texttt{Increment(T, left, right), Decrement(T, left, right)} - Given a counter-clockwise circle segment defined by its two end-directions (\texttt{left} and \texttt{right}), increment or decrement the \texttt{Depth} of all vectors within the circle segment.
\item \texttt{Insert(T, $q'$, $k'$, to)} - Given a new vector $q'$, the number of data points $k'$ in the half-plane corresponding to $q'$, and a boolean value \texttt{to} recording if $q'$ points towards a data point, update the data structure to include $q'$.
\item \texttt{Remove(T, $q'$)} - Given a pointer to an existing vector $q'$, delete $q'$ from the data structure.
\item \texttt{MinDepth(T)} - Return the minimum number of points in a meaningful half-plane associated with $p$.
\end{itemize}

\subsubsection{Detailed Operations}

In our implementation, the data structure is an augmented dynamic balanced binary search tree representing the meaningful half-planes associated with $p$.  The leaves of this tree correspond to the meaningful half-planes and therefore the height of the tree will be $O(\log n)$.  Each meaningful half-plane will be represented in a leaf by its corresponding point on the unit circle in the range $[0,2\pi)$, i.e., the counter-clockwise angle from the positive horizontal vector to the vector associated with the half-plane.  The leaves of the tree are sorted by their direction in the interval $[0,2\pi)$.

The leaves of the tree are augmented with three fields: the \texttt{direction}, \texttt{Depth}, and \texttt{to} fields.  The \texttt{direction} is the value in $[0,2\pi)$ corresponding to the vector associated with the meaningful half-plane.  The \texttt{Depth} is the number of data points in the represented meaningful half-plane.  The \texttt{to} is a boolean value that is true if the vector corresponding to the meaningful half-plane points towards a data point and false if it points away.  The internal nodes of the tree are augmented with three fields: the \texttt{minDepth}, \texttt{maxDepth}, and \texttt{subtree-addition} fields.  The \texttt{minDepth} and \texttt{maxDepth} fields of a node contain the minimum and maximum \texttt{Depth} values of all leaves in its subtree.  The \texttt{subtree-addition} value is a number by which all of the leaves of the node should be incremented or decremented.  The actual updates are deferred until a traversal of this tree is necessary.

To perform an \texttt{Increment(T, left, right)} or \texttt{Decrement(T, left, right)} the \texttt{Depth} values for leaves or the \texttt{subtree-addition} values for internal nodes are incremented or decremented to represent the change in all leaves in a counter-clockwise direction from \texttt{left} to \texttt{right}.  This can be performed by traversing the paths from the leftmost leaf to the rightmost leaf of the interval from \texttt{left} to \texttt{right}.  While traversing these paths, the appropriate \texttt{Depth} or \texttt{subtree-addition} values of the children of this path between \texttt{left} and \texttt{right} are incremented or decremented.  In addition, the \texttt{minDepth} and \texttt{maxDepth} fields of all ancestors of $q'$ may be adjusted as well, if necessary.

An \texttt{Insert(T, $q'$, $k'$, to)} or \texttt{Remove(T, $q'$)} can be performed with a standard dynamic binary tree insert or delete along with an adjustment of the \texttt{minDepth} and \texttt{maxDepth} fields of the ancestors of $q'$, if necessary.  This adjustment can be performed by traversing a path from the root of the tree to $q'$.  The \texttt{MinDepth(T)} can be performed by examining the root of the tree: the \texttt{minDepth} value of the root maintains the half-space depth of $p$, i.e., the half-space depth of $p$ is $(\texttt{minDepth}(\text{root})-1)/n$.

Using these operations, an insertion or deletion of a data point $q$ is straightforward.  Let $T$ be the tree and $q_{to}$ be the vector from $p$ pointing towards $q$ and assume that the angle it makes counter-clockwise from the positive horizontal vector is $\gamma$.  In addition, let $q_{away}$ be the vector from $p$ pointing away from $q$.  The half-planes that contain $q$ are those with direction in the interval $[\gamma,\gamma+\pi]$ (or $[\gamma,2\pi)\cup[0,\gamma-\pi]$ if $\gamma\geq\pi$).  To insert $q$, first perform \texttt{Increment(T,$\gamma$,$\gamma+\pi$)} (or the two increments on the intervals described above) and then perform \texttt{Insert(T, $q_{to}$, $k$, true)} and \texttt{Insert(T, $q_{away}$, $n-k+2$, false)}, where $k$ is the number of points in the half-plane corresponding to $q_{to}$ and can be computed easily by considering the \texttt{Depth} and \texttt{to} fields of a neighbor of $q_{to}$ in the tree.  To delete $q$, first perform \texttt{Decrement(T,$\gamma$,$\gamma+\pi$)} (or the two decrements on the intervals described above) and then perform \texttt{Remove(T, $q_{to}$)} and \texttt{Remove(T, $q_{away}$)}.

The tree, which stores the $2(n-1)$ vectors corresponding to meaningful half-planes with respect to $p$, has linear structure and all operations require logarithmic time.  Since only a single region is incremented or decremented, the dynamic tree on the set of vectors allows associated values to be updated in logarithmic time.

\begin{theorem}
An algorithm exists that can dynamically maintain the half-space depth of a data point relative to a dynamic set of $n$ points in $O(\log n)$ time per operation and linear space.
\end{theorem}
\begin{proof}
To maintain the half-space depth of a single point $p$ the \texttt{minDepth} field of the root needs to be considered after every update.  The depth of $p$ is $(\texttt{minDepth}(root)-1)/n$.
\end{proof}

\section{Dynamically Updating Local Cover Contours}\label{sec:localcontours}
In this section, we augment the data structure from Section 3 to maintain the defining half-planes for the point $p$.  These half-planes provide a witness to the depth of $p$.  In addition, knowing the defining half-planes for $p$ provides a constant time answer for the depth change query: if $q$ is inserted into or removed from the data set, will the half-space depth of $p$ change?  We begin by discussing six properties of the half-space depth cover contours and meaningful half-planes with respect to the data point $p$.  After this, we analyze various cases of insertion or deletion with respect to the point $p$.  These cases prove Theorem \ref{thm:only-one}, a theorem of independent interest, that shows that defining half-planes cannot change too much after an insertion or deletion.  Since this section deals with the defining half-planes, which are not defined for degenerate contours, all of our data points must lie on a non-degenerate contour both {\em before} and {\em after} an update.  At the end of this section, we describe a test which can determine if a point is on a degenerate contour: it amounts to failing the counting property, Property \ref{prop:counting}, below.

\subsection{Properties of the Dynamic Contours}
\label{sec:properties}
The transformation in Section 3 allows us to present 6 important structural properties for the meaningful half-planes.  These properties are the keys to our algorithm for updating the defining half-planes of individual points.  Because these properties are technical, we will fix some notation in this section.  The {\em angle of a vector} used to represent a half-plane is measured counterclockwise from the positive $x$-axis.  The notation $E_{l_1}$ will represent the {\em cover-based contour edge} that lies on the line $l_1$.  A half-plane will be written $H_l$ where $l$ is the {\em boundary line of this half-plane}; although this notation is not unique, the half-plane chosen will be clear from context.  A vector will be written as $v_{H_l}$, where $H_l$ is the half-plane represented by the vector, or as $q_{to}$ and $q_{away}$ when the {\em vectors are associated with the data point $q$}, see Figures \ref{fig:contours2}(a) and \ref{fig:trans}.  For these properties, w.l.o.g., assume that $p$ is the origin.
\begin{property}
{\bf Convexity property:}
\label{prop:convexity}
This property simplifies the following arguments and notation by fixing a standard orientation for a data point and its associated contour: Every half-space depth contour is convex, constraining the angle between the two vectors associated with the defining half-planes, $H_{l_1}$ and $H_{l_2}$.
If edge $E_{l_2}$ is adjacent to $E_{l_1}$ clockwise along the contour of depth $k/n$ with $v_{H_{l_1}}$ pointing towards the positive $x$-axis and if the interior of the contour lies immediately above the
positive $x$-axis, then $v_{H_{l_2}}$ points into the lower semicircle and
the angle of $v_{H_{l_2}}$ is $\phi\in(\pi,2\pi)$, see Figure \ref{fig:contours2}(a).
\end{property}
\begin{property}
\label{prop:insertion}
{\bf Insertion Property:}
This property describes the updates needed as points are added to or deleted from the data set: When a data point $q$ is added to (deleted from) the data set, the number of data points increase (decrease) by one in all half-planes that contain $q$: those whose vectors lie in counter-clockwise interval from $q_{to}$ to $q_{away}$.  If $q_{to}$ has angle $\gamma$ then the affected half-planes have angles $[\gamma,\gamma+\pi]$, see Figure \ref{fig:vectors}(b). (If $\gamma\geq\pi$, then, for the computations below, it is easier to write this interval as $[\gamma,\pi)\cup[0,\gamma-\pi]$.)
\end{property}
\begin{figure}[hbt]
\begin{center}
\begin{tabular}{c@{\hspace{.6in}}c}
\epsfysize=1.7in \epsffile{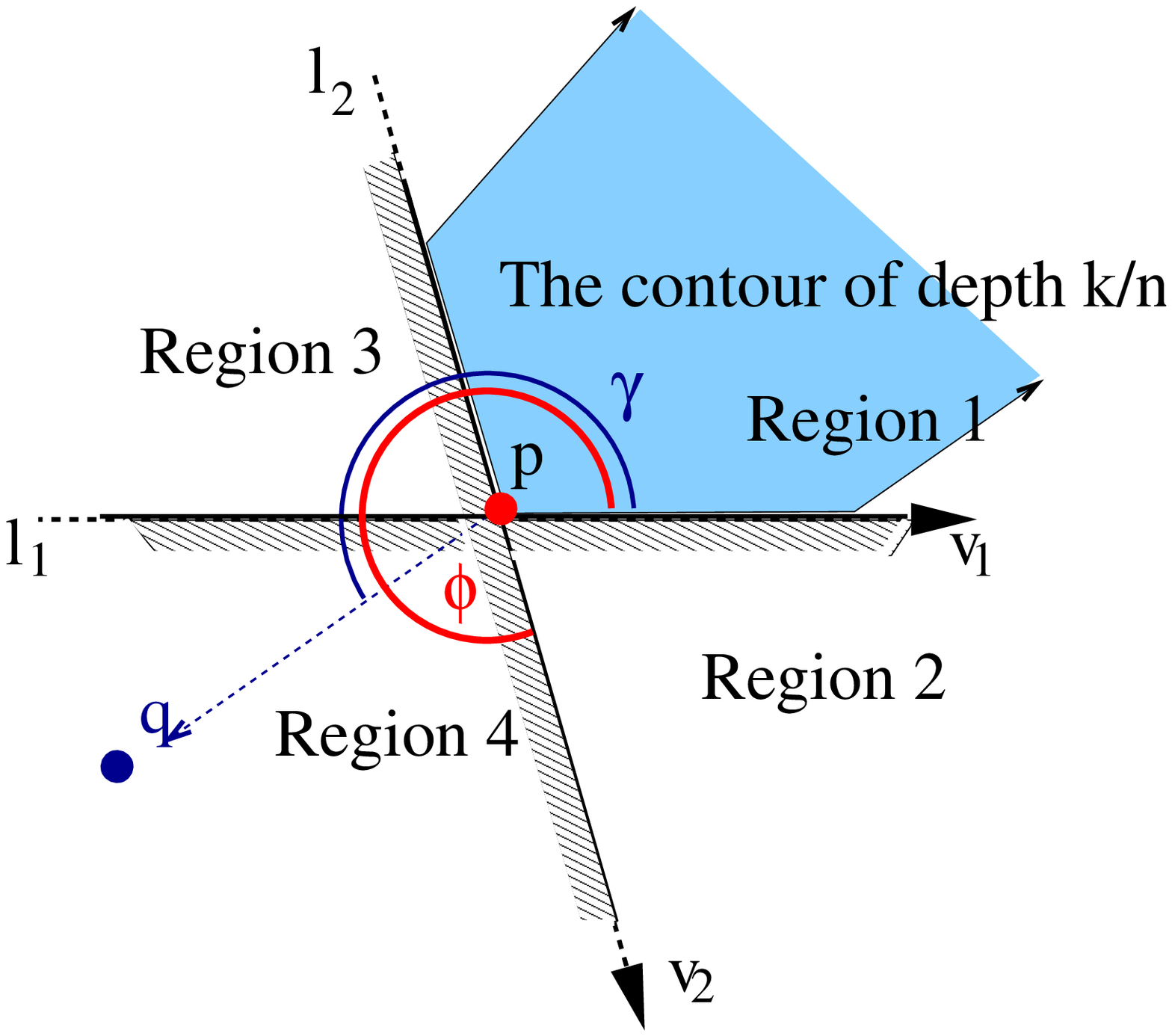}&\epsfysize=1.7in \epsffile{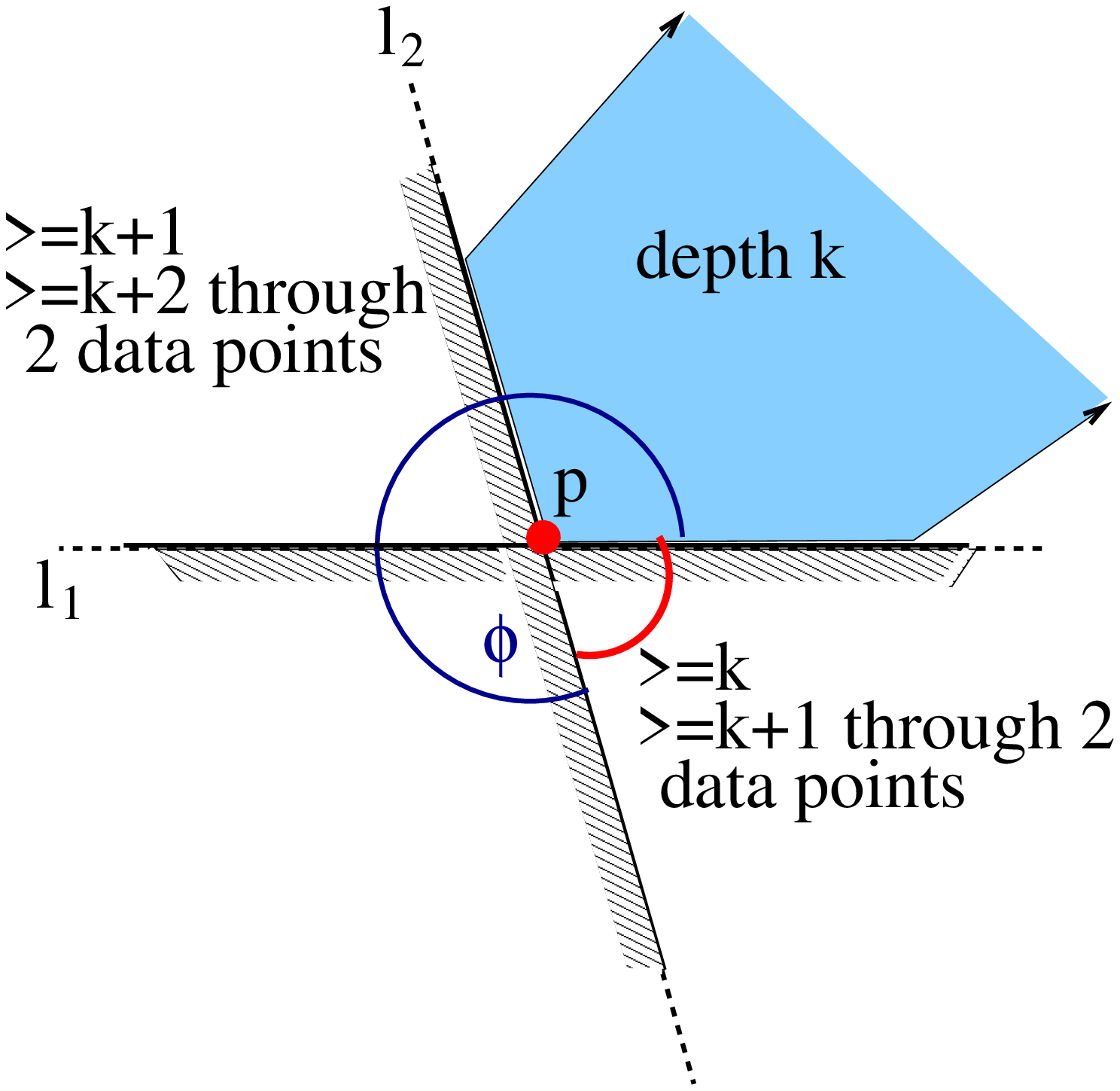}\\
(a)&(b)
\end{tabular}
\end{center}
\caption{(a) The convexity property: The four regions defined by the defining lines for $p$, and the angle to $q$.  The vectors $v_1$ and $v_2$ are the vectors associated with the defining half-planes, and $l_1$ and $l_2$ are the defining lines.  In addition, this contour is in the standard position by the convexity property since $v_1$ points along the positive $x$-axis and $v_2$ points below the $x$-axis. (b) The counting property: The number of points in regions described in the counting and converse counting properties.  For example: If the vector associated with a half-plane points into the interval $(0,\phi)$, then the half-plane contains at least $k+1$ data points.  If the vector associated with a half-plane points into the interval $(\phi,2\pi)$ and its bounding line passes through two data points, then the half-plane also contains at least $k+1$ data points.}
\label{fig:contours2}\label{fig:counting2}
\end{figure}
\begin{property}
{\bf Counting Property:}
\label{prop:counting}
This property describes the minimal number of data points in any half-plane with respect to the local depth contour at a data point: Let $H_m$ be a half-plane where $m$ passes through data point $p$ with half-space depth $k/n$.  Then, if $H_m$ intersects the interior of the cover-based contour of depth $k/n$, then it must contain $>k$ points and if it does not intersect the interior of the cover-based contour, then it must contain $\geq k$ points, see Figure \ref{fig:counting2}(b).  In particular, assume that the defining half-planes for $p$ are $H_{l_1}$ and $H_{l_2}$ and that they satisfy the orientation with angle $\phi$ as described in the convexity property, Property \ref{prop:convexity}.  For any unit vector $v_{H_m}$ with angle $\beta$, consider the number of data points in its the associated half-plane.
\begin{itemize}
\item If $0<\beta<\phi$, then the half-plane $H_m$ intersects the interior of the cover-based contour of depth $k/n$ and must contain $>k$ data points; if $m$ passes through $p$ and a second data point, then $H_m$ must contain $>k+1$ data points. Otherwise (for both cases) there would exist a point (witnessed possibly via a rotation) in the interior of the contour of depth $k/n$ with depth less than $k/n$.
\item If $0=\beta$ or $\phi=\beta$ then $H_m$ is a defining half-plane for $p$ and contains exactly $k+1$ points.
\item If $\phi<\beta<2\pi$, then $H_m$ contains $\geq k$ data points; if $m$ passes through $p$ and another data point, then $H_m$ contains $\geq k+1$ data points (otherwise $p$ would have depth $<k/n$).
\end{itemize}
\end{property}

\begin{property}
{\bf Converse counting property:}
\label{prop:reverse-counting}
This property shows that if the half-planes whose bounding line passes through the data point $p$ follow the conclusions of the counting property, then the depth of $p$ must be $k/n$: Assume that there is a continuous region on the unit circle, $B$, containing vectors whose associated half-planes have $\geq k+1$ data points; the half-planes defined by $p$ and another data point whose vectors are in $B$ contain $\geq k+2$ points.  The half-planes whose associated vectors are on the boundary of $B$ must have 2 points on their bounding line and their half-planes contain exactly $k+1$ points.  The half-planes whose vectors are in the complementary region $B^C$ contain $\geq k$ points; the half-planes defined by $p$ and another data point whose vectors are in $B^C$ contain $\geq k+1$ points.  Then point $p$ must have depth $k/n$.  Since every half-plane contains $k$ points, the half-space depth of $p$ is at least $k/n$.  Let $H_{m}$ be one of the half-planes forming the boundary of $B$ and $q$ the data point not equal to $p$ on $m$.  By a rotation of $H_m$, a second half-plane can be found which contains all data points except for $q$.  This new half-plane will therefore contain $p$ and $k$ data points.  Therefore, the depth of $p$ is $k/n$, see Figure \ref{fig:counting2}(b).
\end{property}
\begin{property}
{\bf Location property:}
\label{prop:location}
Given two candidate defining half-planes, i.e., ones which contain the correct number of data points, this property describes the possible positions for the defining half-planes: Given two arbitrary half-planes $H_{m_1}$, $H_{m_2}$ whose bounding lines each pass through point $p$ (of depth $k/n$) and one other data point and which each contain $k+1$ data points, then the vectors for the defining half-planes are restricted to lie in two small intervals, see Figure \ref{fig:counting3}(a).  Assume $v_{H_{m_1}}$ points towards the positive $x$-axis and $v_{H_{m_2}}$ points into the lower semicircle with angle $\beta$.  By the converse counting property, Property \ref{prop:reverse-counting}, the angles of the associated vectors of the {\em defining half-planes} must be in the interval $[0,\beta]$.  Alternately, if an associated vector for a defining half-plane makes an angle in the interval $[\beta-\pi,\pi]$, then the union of this half-plane, $H_{m_1}$, and $H_{m_2}$ is the entire plane, contradicting the assumption that $p$ is on a non-trivial contour.  Therefore, the angles of the associated vectors for the defining half-planes must come from $[0,\beta-\pi)$ and $(\pi,\beta]$, one vector from each interval because, by the counting property, Property \ref{prop:counting}, $H_{m_1}$ and $H_{m_2}$ must lie in Region 2 from Figure \ref{fig:contours2}(a).
\end{property}

\begin{figure}[hbt]
\begin{center}
\begin{tabular}{c@{\hspace{.6in}}c}
\epsfysize=1.7in
\epsffile{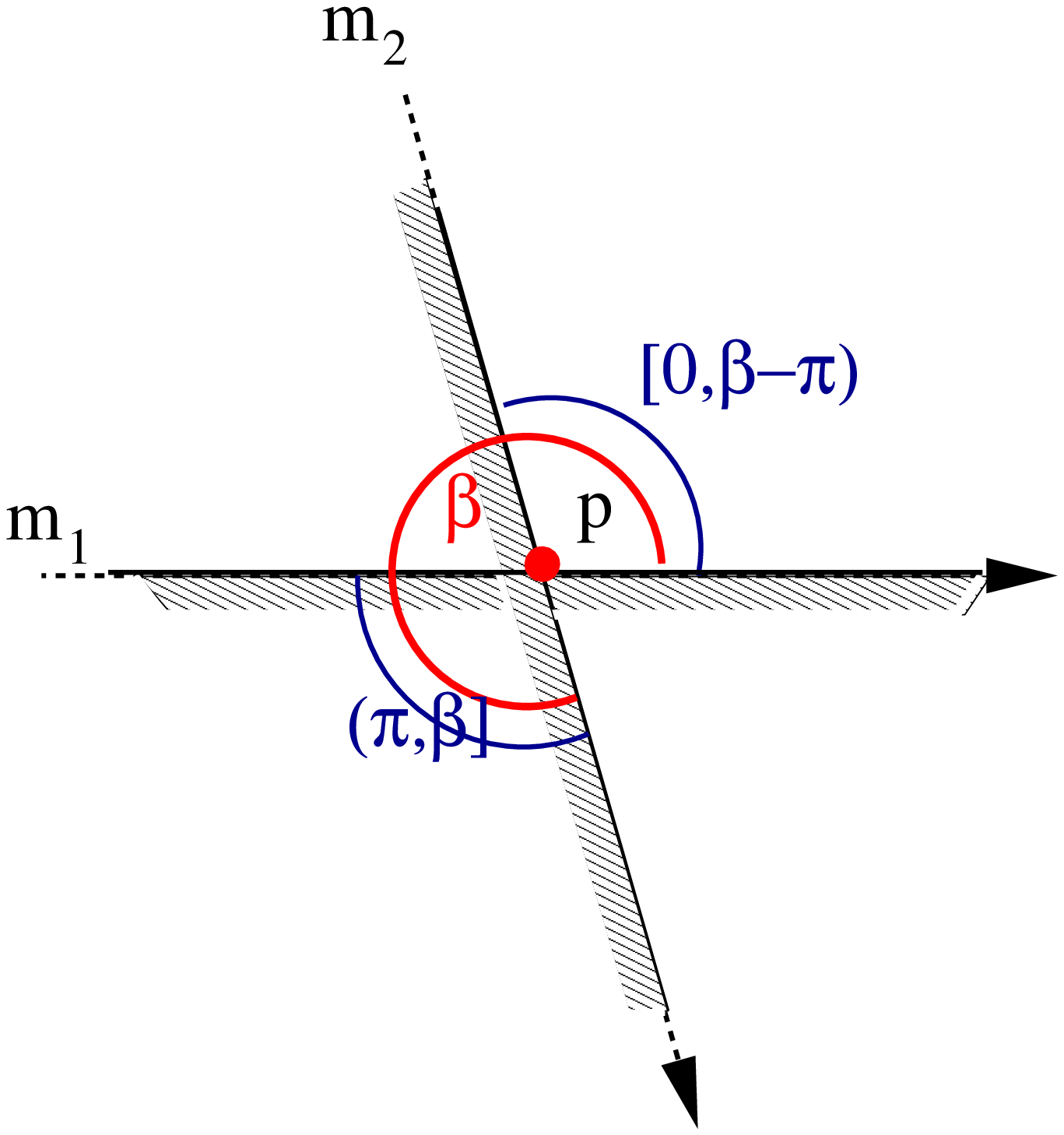}&
\epsfysize=1.7in
\epsffile{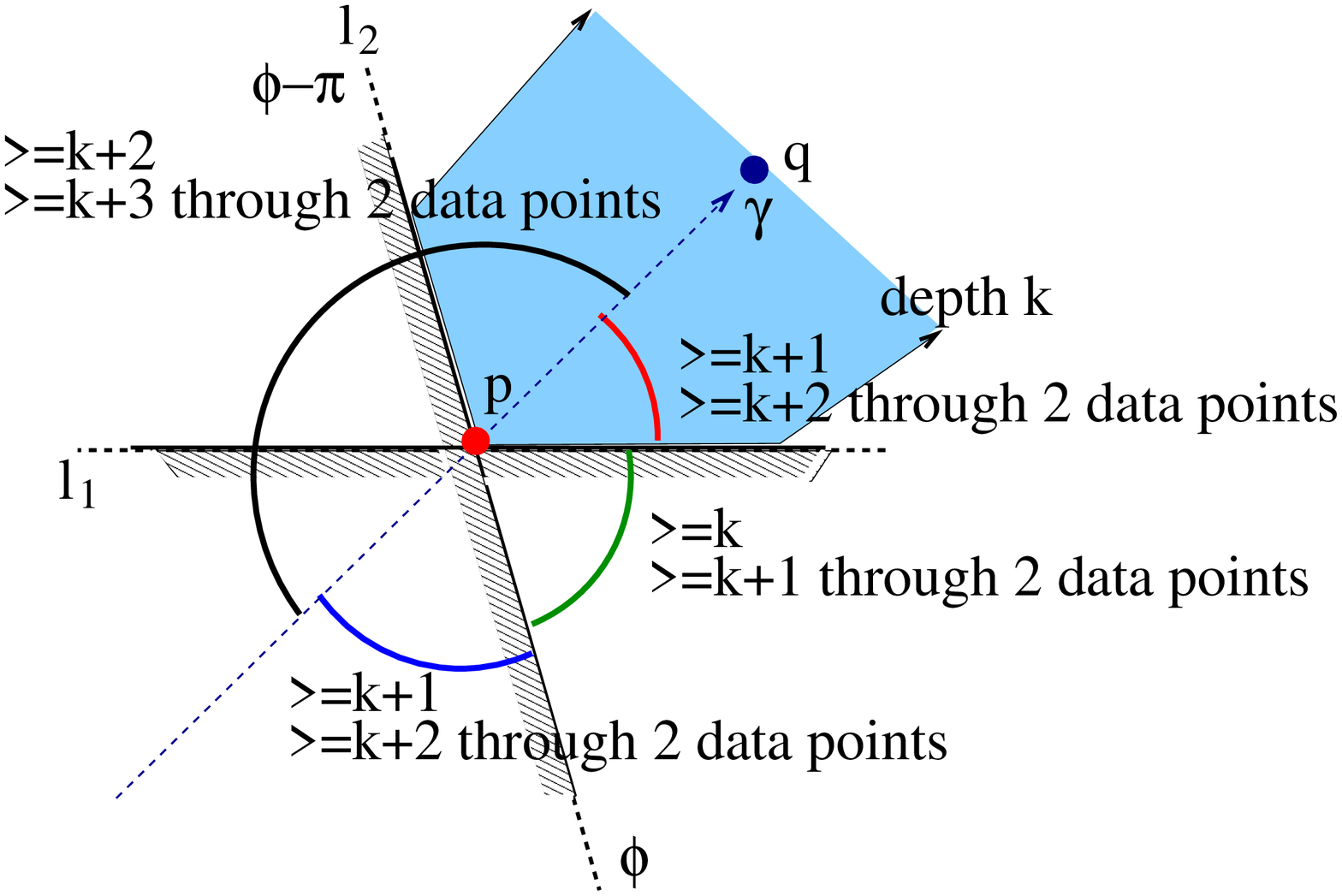}\\
(a)&(b)
\end{tabular}
\end{center}
\caption{(a) The location property: The two half-planes $H_{m_1}$ and $H_{m_2}$ are candidate defining half-planes because they include the correct number of data points.  The actual defining half-planes must lie in the two smaller intervals $[0,\beta-\pi)$ and $(\pi,\beta]$ (one in each) to ensure that the counting property, Property \ref{prop:counting}, remains valid. (b) Case 1: The counts of all half-planes when point $q$ is inserted into region 1.  The contour shown is the old contour of depth $k/n$.  The depth of $p$ and the defining half-planes for $p$ remain unchanged by the counting property, Property \ref{prop:counting}.}
\label{fig:counting3}\label{fig:caseone}
\end{figure}

\begin{property}
\label{prop:nestedness}
{\bf Nestedness property:}
This property shows that the contour regions can only expand when a point is added to the data set: When a point $q$ is added to the data set, the half-space depth cover-based contour of depth $k/n$ relative to $\mathcal{F}_n$ is nested in the half-space depth cover-based contour of depth $k/(n+1)$ relative to $\mathcal{F}_n \cup \{q \}$.  When $q$ is added, the depth of all half-planes either increases or remains the same.  Thus, the depth of every point either increases or remains the same.  Assume $p$ lies on a nontrivial contour with defining half-planes $H_{l_1}$ and $H_{l_2}$ and the depth of $p$ remains unchanged when $q$ is inserted.  In addition, assume that the defining half-planes for $p$ satisfy the orientation described in the convexity property, Property \ref{prop:convexity}.  Then, the angles of the associated vectors for the defining half-planes after $q$ is inserted must be in the interval $\{0\}\cup[\phi,2\pi)$.
\end{property}
\subsection{Possible Cases}
\label{sec:cases}
Consider a data set $\mathcal{F}_n$ before and after an update step, where point $q$ is either inserted into or deleted from the data set.  Assume point $p$ has depth $k/n$ and lies on a nontrivial contour
before and after the update.  In addition, assume that the defining half-planes for $p$ with respect to $\mathcal{F}_n$ are $H_{l_1}$ and $H_{l_2}$ satisfying the orientation described in the convexity property, Property \ref{prop:convexity}.  Let the two defining half-planes for $p$ after the update (in $\mathcal{F}_n\cup\{q\}$ or $\mathcal{F}_n\setminus\{q\}$) be $H_{m_1}$ and $H_{m_2}$.
In this case, lines $l_1$ and $l_2$ divide the plane into 4 regions, see Figure \ref{fig:contours2}(a).
Call region  $(H_{l_1})^C\cap(H_{l_2})^C$ {\em region 1}.
Call $H_{l_1}\cap (H_{l_2})^C$ {\em region 2} and  $(H_{l_1})^C \cap H_{l_2}$ {\em region 3}.  Call $H_{l_1}\cap H_{l_2}$ {\em region 4}.
Let $v=q_{to}$ be the unit vector from $p$ towards $q$ and let $\gamma$ be $v$'s angle.

Consider the possible cases created by inserting $q$ relative to the four regions.  Point $q$ can be inserted only into the interior of one of the four regions (inserting into the boundary will violate general position) and deleted from the interior of any of the four regions or a defining edge.
We first present an analysis of the cases created when a point $q$ is inserted into the set.

\begin{case}
{\bf Point $q$ is inserted into region 1 (see Figure \ref{fig:caseone}(b)):}
$q$ is inserted into neither of the defining half-planes $H_{l_1}, H_{l_2}$, so $0<\gamma<\phi-\pi$.
Since the number of data points contained in $H_{l_1}$ and $H_{l_2}$ does not change, the depth of $p$ remains $k/(n+1)$ by the converse counting property, Property \ref{prop:reverse-counting}.
By the nestedness property, Property \ref{prop:nestedness},
the new defining-half planes
must be in the interval $\{0\}\cup[\phi,2\pi)$. By the location property,
Property \ref{prop:location},
the defining half-planes must also be in the intervals $[0,\phi-\pi)$ and $(\pi,\phi]$.  The intersection of the two sets
is $\{0,\phi\}$.
Since after the update $p$ lies on a nontrivial contour, then $H_{m_1}$ and $H_{m_2}$ are distinct
so the defining half-planes for $p$ remain unchanged.
\end{case}
\begin{figure}[hbt]
\begin{center}
\epsfysize=1.7in
\epsffile{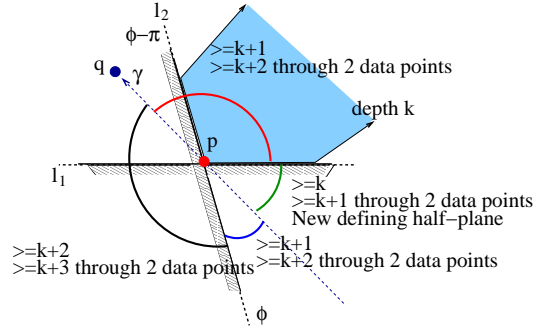}
\end{center}
\caption{Case 2: The counts of all half-planes when a point $q$ is inserted into region 3.  The contour shown is the old contour of depth $k/n$.  The depth of $p$ remains unchanged, but the old defining half-plane with angle $\phi$ is no longer a defining half-plane.  It is the first half-plane containing $k+1$ points and whose boundary passes through two data points in the indicated region when traversed clockwise from the $\gamma+\pi$ angle.}
\label{fig:casetwo}
\end{figure}
\begin{case}
\label{case:insert2}
{\bf Point $q$ is inserted into region 2 or 3 (see Figure \ref{fig:casetwo}):} $q$ is inserted into exactly one defining half-plane.
Assume $q$ is inserted into region 3, thus $H_{l_2}$ contains an extra point, $H_{l_1}$ does not contain an extra point, and $\phi-\pi<\gamma<\pi$.  By the insertion property, Property \ref{prop:insertion},
all half-planes with angles between $\gamma$ and $\gamma+\pi$ include an additional point.  $H_{l_l}$ will not contain an additional point, so the depth of $p$ will remain $k/(n+1)$, but $H_{l_2}$ will include an extra point so it can no longer be a defining half-plane for $p$ by the counting property, Property \ref{prop:counting}.
By the nestedness property, Property \ref{prop:nestedness},
the angle of vector $v_{H_{m_2}}$
must be in the interval $U_1  =\{0\}\cup[\phi,2\pi)$.  By the insertion property, Property \ref{prop:insertion},
every defining plane whose vector has an angle in the interval $U_2 = [\gamma,\gamma+\pi]$ includes an additional point, therefore, by combining the counting property, Property \ref{prop:counting}, the angle for $v_{H_{m_2}}$ must be in the interval $U_1 \setminus U_2 = [\gamma+\pi,2\pi]$.
By the counting property, Property \ref{prop:counting},
$v_{H_{m_2}}$ must be the first vector encountered when traversing the vectors starting from $v_{H_{l_2}}$ counter-clockwise, whose bounding line passes through two points and associated half-plane contains $k+1$ data points.

Assume $v_{H_{m_2}}$ has
angle $\beta$
and consider the location of the second defining half-plane $H_{m_1}$:
By the location property, Property \ref{prop:location},
the associated vectors to the new defining half-planes must be in the intervals $V_1=[0,\beta-\pi)$ and $V_2=(\pi,\beta]$, so $v_{H_{m_1}}$ must be in the interval $[0,\beta-\pi)$ since $v_{H_{m_2}}\in(\pi,\beta]$.  However, by the nestedness property, Property \ref{prop:nestedness},
the angles of the associated vectors for each of the defining half-plane must also be in $U_1$, but $U_1\cap V_1=\{0\}$.
Therefore, the associated vector for $H_{m_1}$ must have
angle $0$,
and $H_{m_1}=H_{l_1}$.
\end{case}

\begin{figure}[hbt]
\begin{center}
\epsfysize=1.7in
\epsffile{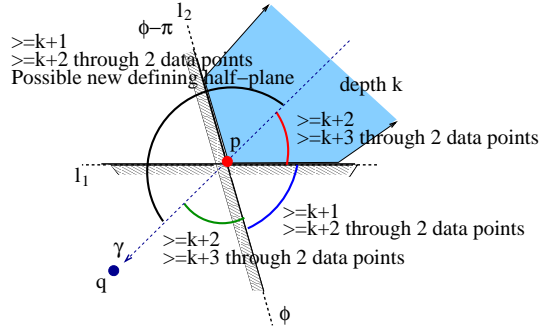}
\end{center}
\caption{Case 3: The counts all half-planes when a point $q$ is inserted into region 4.  The contour shown is the old contour of depth $k/n$.  The depth of $p$ increases by 1 and at most one half-plane can move into the indicated region.  The new defining half-plane will be the first half-plane whose angle is either counter-clockwise from $\pi$ or clockwise from $\phi-\pi$ which contains $k+1$ points and whose boundary passes through two data points.}\label{fig:casethree}
\end{figure}

\begin{case}
\label{case:insert4}
{\bf Point $q$ is inserted into region 4 (see Figure \ref{fig:casethree}):} $q$ is inserted into both defining half-planes, thus $\pi<\gamma<\phi$.  By the insertion property, Property \ref{prop:insertion}, every half-plane whose associated vector has angle in the interval $[\gamma, \gamma+\pi]$ gains an additional point.  Since $\gamma>\pi$, we safely replace this interval with $U=[\gamma,2\pi) \cup [0,\gamma-\pi]$.
Combining these two ideas, we note that, in particular, every half-plane in the interval $\{0\}\cup[\phi,2\pi)$ gains an additional point.  By the counting property, Property \ref{prop:counting}, every half-plane whose vector has angle in the complement of $U$ already contains at least $k+1$ data points and every half-plane whose vector has angle in $U$, which previously contained at least $k$ data points, now contains at least $k+1$.  Therefore the depth of $p$ increases to $(k+1)/(n+1)$.

The two defining half-planes of $p$ relative to $\mathcal{F}_n$, $H_{l_1}$ and $H_{l_2}$ contain $k+2$ data points each.  By the location property, Property \ref{prop:location}, the angles of the associated vectors $v_{H_{m_1}}, v_{H_{m_2}}$ must be in the intervals $[0,\phi-\pi)$ and $(\pi,\phi]$.  Assume that $H_{l_2}$ is no longer a defining half-plane for $p$, then the angle of $v_{H_{m_2}}$ is in the interval $(\pi,\phi)$.  By the counting property,
Property \ref{prop:counting}, combined with the insertion property, Property \ref{prop:insertion}, every half-plane whose vector has angle in $(\gamma,\phi)$ contains at least $k+2$ data points and whose bounding line passes through two data points must contain $k+3$ data points and cannot be a defining half-plane for $p$.  Thus, the angle for $v_{H_{m_2}}$ must be in the interval $(\pi,\gamma]$.  Assume $v_{H_{m_2}}$ has angle $\beta$.  By the counting property, Property \ref{prop:counting}, $v_{H_{m_2}}$ must be the first vector counter-clockwise from $\pi$ to $\gamma$ where $m_2$ passes through two data points and $H_{m_2}$ contains $k+2$ data points.

Now, assume that $H_{l_1}$ is also no longer a defining half-plane for $p$.  Then by the location property, Property \ref{prop:location}, the angle for $v_{m_1}$, must be in the interval $(0,\beta-\pi)$.  By construction, $\beta\leq\gamma$, so $\beta-\pi\leq\gamma-\pi$. By the counting property, Property \ref{prop:counting}, combined with the insertion property, Property \ref{prop:insertion}, every half-plane whose vector has angle in $(0,\gamma-\pi)$ contains at least $k+2$ data points and whose bounding line passes through two data points must contain $k+3$ data points and cannot be a defining half-plane for $p$.  However, since $(0,\beta-\pi)\subseteq(0,\gamma-\pi)$ there is a contradiction since the angle for $v_{m_1}$ must be in that interval.  Thus $v_{m_1}=v_{l_1}$ and $H_{m_1} = H_{l_1}$.

The case where $H_{l_1}$ is no longer a defining half-plane is handled similarly.
\end{case}

To describe the cases dealing with the removal of point $q$ we will use the following theorem of independent interest that shows that the defining half-planes cannot change by much after an update.

\begin{theorem}
\label{thm:only-one}
Let $p$ be a data point which lies on a nontrivial half-space depth contour with respect to the data set $\mathcal{F}_n$ and assume that after a point $q$ is inserted into (deleted from) $\mathcal{F}_n$, $p$ once again lies on a nontrivial contour.
Then at least one of the two defining half-planes of $p$ (with respect to $\mathcal{F}_n$) remains unchanged with respect to the new data set $\mathcal{F}_n \cup\{q\}$ ($\mathcal{F}_n \setminus\{q\}$).
\end{theorem}
\begin{proof}
By cases 1-3, when $q$ is inserted into $\mathcal{F}_n$, one of the defining half-planes remains unchanged so the theorem is proved for the cases of insertions.

The claim is true when $q$ is inserted, but assume, for contradiction, that it is false when $q$ is deleted.  Consider first deleting and then reinserting $q$.  After these two updates, the data set returns to its original state and the defining half-planes for $p$ must be unchanged.  Moreover, by above, the reinsertion of $q$ can only change one of the defining half-planes for $p$ with respect to $\mathcal{F}_n\setminus\{q\}$.  If the deletion of $q$ changed both defining half-planes, it is impossible for the reinsertion of $q$ to return the defining half-planes to their original state.  This is a contradiction and therefore the theorem is proved for the cases of deletions.
\end{proof}

We now analyze the cases where the point $q$ is deleted.  In particular, Theorem \ref{thm:only-one} is non-constructive for deletions, but the constructions themselves are necessary for the description of an algorithm.  For most of these cases, we determine the new defining half-planes by performing an insertion backwards and considering which type of insertion could result in the desired situation.

\begin{figure}[hbt]
\begin{center}
\epsfysize=1.7in
\epsffile{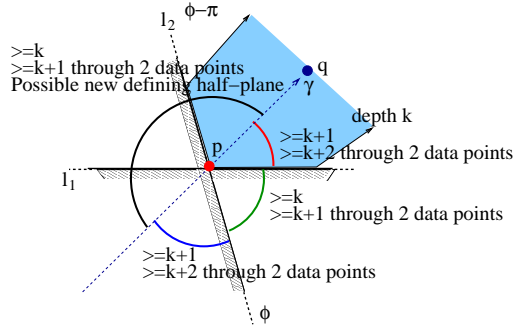}
\end{center}
\caption{Case 4: The counts of all half-planes when point $q$ is deleted from region 1.  The contour shown is the old contour of depth $k/n$.  The depth of $p$ remains unchanged and at most one half-plane will move into the indicated region.  The new defining half-plane will be the first half-plane whose angle is either clockwise from $\phi-\pi$ or counter-clockwise from $\pi$ which contains $k+1$ points and whose boundary passes through two data points.}
\label{fig:casefour}
\end{figure}

\begin{case}
{\bf Point $q$ is deleted from region 1 (see Figure \ref{fig:casefour}):}
$q$ is deleted from neither of the defining half-planes, thus $0<\gamma<\phi-\pi$.
Consider the counting property, Property \ref{prop:counting}, combined with the insertion property, Property \ref{prop:insertion}, and consider a half-plane $H_l$ whose associated vector $v_{H_l}$ has angle $\beta$.
If $0<\beta<\gamma$ or $\gamma+\pi<\beta<\phi$, then $q$ does not affect this half-plane and $H_l$ contains at least $k+1$ data points both before and after $q$ was removed and contains at least $k+2$ data points if its bounding line passes through two data points.
If $\gamma\leq\beta\leq\gamma+\pi$, then $q$ is contained in $H_l$ and the number of data points in $H_l$ decreases by one.  Before $q$ was deleted $H_l$ contained at least $k+1$ data points and contained at least $k+2$ data points if its bounding line passed through two data points.   After $q$ was deleted, there are at least $k$ points in $H_l$ and at least $k+1$ data points in a half-plane whose bounding line passes through two data points.
If $\phi<\beta<2\pi$, then $q$ does not affect this half-plane and $H_l$ contains at least $k$ data points and $k+1$ data points if its bounding line passes through two data points before and after the deletion.
Finally, if $\beta = 0$ or $\alpha$ ($H_l$'s associated vector $v_{H_l}$ is either $v_{H_1}$ or $v_{H_2}$), then $q$ does not affect this half-plane and $H_l$ contains exactly $k+1$ before and after the deletion.
The number of points in some half-planes has decreased, but the half-space depth of $p$ is still at least $k/(n-1)$ so the depth of $p$ remains unchanged.

The defining half-planes may change location. By the location property, Property \ref{prop:location}, the associated vectors to the new defining half-planes must be in the regions $[0,\phi-\pi)$ and $(\pi,\phi]$.  By the discussion above, we can shrink these regions to $(\gamma,\phi-\pi)\cup\{0\}$ and $(\pi,\gamma+\pi)\cup\{\phi\}$.  By the counting property, Property \ref{prop:counting}, to locate these vectors one needs to traverse the region from $\phi-\pi$ towards $\gamma$ and from $\pi$ towards $\gamma+\pi$ until the first vector whose associated half-plane passes through two data points and contains $k+1$ data points is located in each direction.  If none exist, then the defining half-planes remain unchanged.  By Theorem \ref{thm:only-one}, at most one defining half-plane will change its location.
\end{case}

\begin{figure}[hbt]
\begin{center}
\epsfysize=1.7in
\epsffile{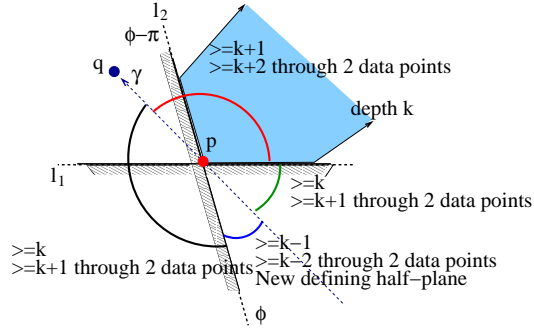}
\end{center}
\caption{Case 5: The counts of all half-planes when point $q$ is deleted from region 3.  The contour shown is the old contour of depth $k/n$.  The depth of $p$ decreases and the half-plane with angle 0 is no longer a defining half-plane.  It is the first half-plane containing $k$ points and whose boundary passes through two data points in the indicated region when traversed clockwise from $\gamma+\pi$ to $\phi$.}
\label{fig:casefive}
\end{figure}

\begin{case}
{\bf Point $q$ is deleted from region 2 or 3 (see Figure \ref{fig:casefive}:}
$q$ is deleted from one of the defining half-planes.
Assume w.l.o.g. that $q$ is deleted from region 3.
For the insertion of data point $q$ to end up in region 3 (or 2) it must be inserted into region 4 (Case \ref{case:insert4}).
Insertion into region 4 increases the depth of $p$ and may move one of the defining lines into the inside of the old contour.
Therefore a deletion will decrease the depth of $p$.  After the deletion, $H_{l_1}$ will still contain $k+1$ data points and have two data points on its boundary; also, $H_{l_2}$ will now contain only $k$ data points and have two data points on its boundary.  Therefore, $H_{l_1}$ can no longer be a defining half-plane for $p$.  In fact, since $q$ must now correspond to a point in region 4, $H_{l_1}$ must change to a half-plane containing $k$ points and whose boundary includes two points in the interval $(\phi,\gamma+\pi]$.  To locate this vector, traverse the region from $\gamma+\pi$ to $\phi$ in the clockwise direction to find the first vector whose associated half-plane satisfies these conditions.
\end{case}

\begin{figure}[hbt]
\begin{center}
\epsfysize=1.7in
\epsffile{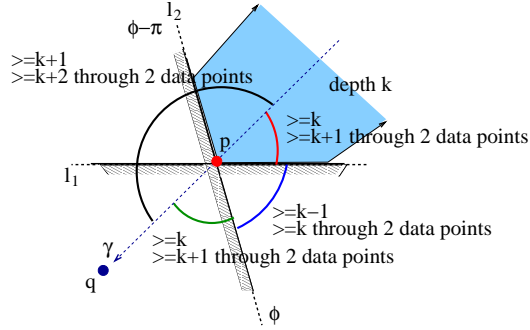}
\end{center}
\caption{Case 6: The counts of all half-planes when point $q$ is deleted from region 4.  The contour shown is the old contour of depth $k/n$.  The depth of $p$ decreases and the defining half-planes for $p$ remain unchanged by the counting property, Property \ref{prop:counting}.}
\label{fig:casesix}
\end{figure}

\begin{case}
{\bf Point $q$ is deleted from region 4 (see Figure \ref{fig:casesix}):}
$q$ is deleted from both defining half-planes $\pi<\gamma<\phi$.
Since both $H_{l_1}$ and $H_{l_2}$ now contain $k$ points and their bounding lines pass through two data points, the depth of $p$ is now $(k-1)/(n-1)$.  By the location property, Property \ref{prop:location}, the new defining half-planes must lie in the intervals $[0,\phi-\pi)$ and $(\pi,\phi]$, with one in each interval.  However, by the counting property, Property \ref{prop:counting}, every half-plane in the interval $(0,\phi)$ whose bounding line passes through two data points contains $\geq k+2$ data points.  After the deletion of $q$, these half-planes contain $\geq k+1$ data points.  Therefore, they cannot be defining half-planes and $H_{l_1}$ and $H_{l_2}$ must remain the defining half-planes.
\end{case}

\begin{figure}[hbt]
\begin{center}
\begin{tabular}{c@{\hspace{.2in}}c}
\epsfysize=1.7in
\epsffile{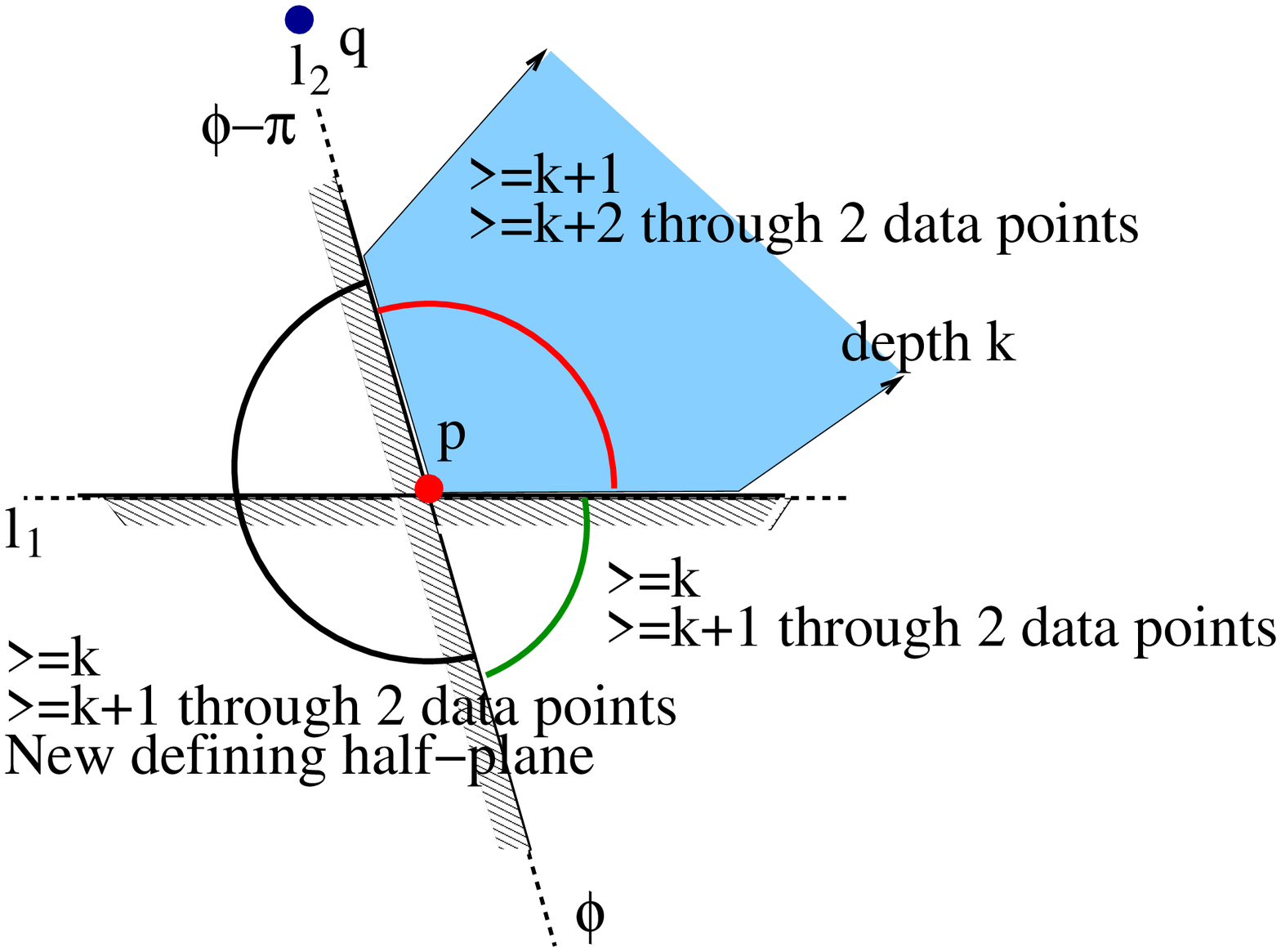}&
\epsfysize=1.7in
\epsffile{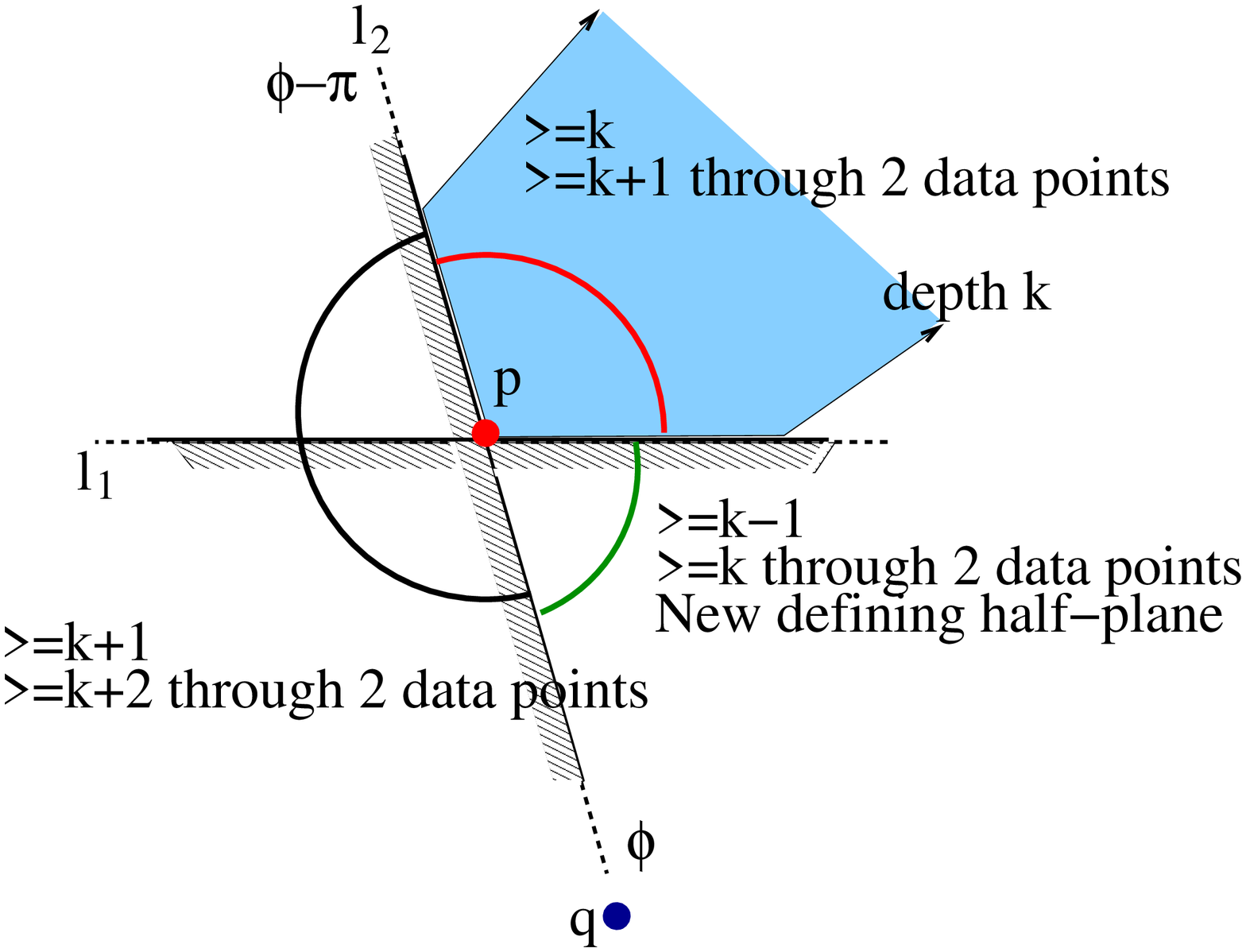}\\
(a)&(b)
\end{tabular}
\end{center}
\caption{Case 7: The counts of all half-planes when point $q$ is deleted from a defining half-plane.  (a) The case where $q$ is in the same direction as the defining edge and is deleted from the boundary between regions 1 and 3.  The half-space depth of $p$ remains unchanged, but the defining half-plane with angle $\phi$ is no longer a defining half-plane.  The new defining half-plane is the first half-plane containing $k+1$ data points and whose boundary passes through two data points in the counter-clockwise direction from $\pi$ to $\phi$. (b) The case where $q$ is in the opposite direction from the defining edge and is deleted from the boundary between regions 2 and 4.  The half-space depth of $p$ decreases, and the defining half-plane with angle $\phi$ is no longer a defining half-plane.  The new defining half-plane is the first half-plane containing $k+1$ data points and whose boundary passes through two data points in the counter-clockwise direction from $\phi$ to $2\pi$.}
\label{fig:caseseven}
\end{figure}

\begin{case}
{\bf Point $q$ is on a defining line (see Figure \ref{fig:caseseven})}
There are two possible cases when $q$ is on a defining line.  The two cases occur if $q$ is in the same direction from $p$ as the cover contour edge, see Figure \ref{fig:caseseven}(a), or if $q$ is in the opposite direction as the contour edge, see Figure \ref{fig:caseseven}(b).  If $q$ lies on the boundary between regions 1 and 2 or regions 1 and 3, then $q$ lies in the direction of the contour edge.  For the insertion of data point $q$ to end up in this position on a definition line, it must be inserted into region 2 or 3 (Case \ref{case:insert2}).  If $q$ lies on the boundary between regions 2 and 4 or regions 3 and 4, then $q$ lies in the opposite direction from the contour edge.  For the insertion of data point $q$ to end up in this position on a defining line, it must be inserted into region 4 (Case \ref{case:insert4}).  In either case, the half-plane including $q$ can no longer be a defining half-plane since its boundary line passes through only one data point.

In first case, assume w.l.o.g. that $q$ lies on the boundary between regions 1 and 3, see Figure \ref{fig:caseseven}(a).  Then $H_{l_2}$ is no longer a defining half-plane.  By undoing an insertion into region 3, Case \ref{case:insert2}, we can locate the new defining half-plane by searching in a counter-clockwise direction from $\pi$ to $\phi$ to find the first vector whose associated half-plane contains $k+1$ data points and whose boundary contains two data points.

In the second case, the depth of $p$ decreases, and assume w.l.o.g. that $q$ lies on the boundary between regions 2 and 4, see Figure \ref{fig:caseseven}(b).  Then $H_{l_2}$ is no longer a defining half-plane.  By undoing an insertion into region 4, Case \ref{case:insert4}, we can locate the new defining half-plane by searching in a counter-clockwise direction from $\phi$ to $2\pi$ to find the first vector whose associated half-plane contains $k$ data points and whose boundary contains two data points.
\end{case}

\subsection{Algorithm}
In this section, we augment the data structure presented in Section 3 to maintain not only the half-space depth of a point, but also the defining half-planes, provided that the data point is not on a degenerate contour.  By inspecting the arguments in Section \ref{sec:cases}, we find that all that is needed is a search ability to find new defining half-planes, this operation is described here.  Note that although all of the arguments in Section \ref{sec:cases} assumed that the current position of the half-planes obeyed the convexity property, Property \ref{prop:convexity}, it is straightforward to reinterpret and rotate the results for defining half-planes in arbitrary positions (e.g., a rotation is just a translation of the unit circle).

\begin{itemize}
\item \texttt{Next(T, $q$, $k$), Prev(T, $q$, $k$)} - Given a set of vectors \texttt{T} emanating from data point $p$, labeled with \texttt{Depth}, a vector $q$, and a depth value $k$, report the first vector in a counter-clockwise or clockwise direction, respectively, after $q$ of \texttt{Depth} $k$ (assuming such a vector exists).
\end{itemize}

\subsubsection{Detailed Operations}
\begin{itemize}
\item To perform a \texttt{Next(T, $q$, $k$)}, first locate the leaf $q$ by searching the tree, and, while searching the tree, update the \texttt{Depth}, \texttt{subtree-addition}, \texttt{minDepth}, and \texttt{maxDepth} fields of the children of the path as the tree is traversed.  Now, traverse this path in the reverse direction.  From bottom to top, every time there is a right child of the path that is not on the path, consider its \texttt{minDepth} and \texttt{maxDepth} values.  If $k$ is contained in this range, descend through this tree taking the path to the left child whenever k is included in the interval formed from its \texttt{minDepth} and \texttt{maxDepth} values and the right child otherwise.  If we reach the root of the tree, without finding appropriate \texttt{minDepth} and \texttt{maxDepth} values, let $left$ be the leftmost leaf in the tree and perform \texttt{Next(T, $left$, $k$)} (this step wraps around the break between $2\pi$ and $0$ in the unit circle).  To perform \text{Prev(T, $q$, $k$)}, perform the same operation as described above in \texttt{Next}, but with the roles of right and left reversed.
\end{itemize}

\begin{theorem}
An algorithm exists that can dynamically maintain the half-space depth of a data point and maintain its defining half-planes, if defined, relative to a dynamic set of $n$ points in $O(\log n)$ time per operation and linear space.
\end{theorem}
\begin{proof}
Maintain the structure presented in Section 3, augmented with the operations from this section and update the defining half-planes using the operations as described in Section \ref{sec:cases}
\end{proof}

\begin{theorem}
If $p$ lies on a non-degenerate contour and the defining half-planes for a point $p$ are known, then the new depth of $p$ can be determined when $q$ is either inserted into or deleted from the data set in $O(1)$ time.
\end{theorem}
\begin{proof}
The cases in Section \ref{sec:cases} describe all possibilities and can be decided in by comparing the point of interest to the defining half-planes.
\end{proof}

Thus far, we have avoided the degenerate case.  It is possible to check whether this case occurs using the \texttt{Next} and \texttt{Prev} operations.  We use the convexity property, Property \ref{prop:convexity}, and the counting property, Property \ref{prop:counting}, to note that if $p$ has depth $k/n$, then all meaningful half-planes with respect to $p$ containing $k+1$ data points and whose boundary contains $2$ data points must have associated vectors lying in a semi-circle.  One easy way to check this is as follows: Divide the unit circle into four intervals: $[0,\pi/2]$, $[\pi/2,\pi]$, $[\pi,3\pi/2]$, and $[3\pi/2,2\pi]$.  Search for the extreme meaningful half-planes in each interval containing $k$ data points.  Then $p$ is on a degenerate contour if and only if these vectors do not lie in a semi-circle.

\section{Dynamically Updating Rank Contours}
In this section, we keep a collection of data structures either from Section 3 or from Section 4 to maintain the depth of each point and then resort the data points according to their new depth.  (There is a choice of data structures because the structures from Section \ref{sec:dynamicdepthpoints} maintain the data necessary for this algorithm.  Since the structures from Section \ref{sec:localcontours} are augmented versions of this data structure, they also maintain the necessary data.)  We begin with a simple observation:

\begin{observation}
\label{obs:depth-change}
When a data point is inserted into or deleted from a data set, the value of $k$ in each point's depth, $k/n$, can increase or decrease by at most 1 step; i.e., for an insertion the new half-space depth is either $k/(n+1)$ or $(k+1)/(n+1)$ and for a deletion, the new half-space depth is either $k/(n-1)$ or $(k-1)/(n-1)$.
\end{observation}

The re-sorting step uses a list of buckets that hold all data points, where bucket $k$ holds the list of data points of depth $k/n$. The list contains buckets from $0$ to $\lfloor n/2\rfloor$, the maximum number of contours for the current data set.  When a new point is inserted or deleted, every point $q$ that changes its depth must be moved: for example, if a data point changes its half-space depth from $k'/n$ to $(k'+1)/(n+1)$, then it is removed from its old bucket ($k'$) and inserted into the new bucket ($k'+1$).  Observation \ref{obs:depth-change} guarantees that the depth of $q$ does not change by more then 1 and so the update takes constant time. The structure guarantees that the data points {\em sorted by depth} can be retrieved in linear time.

To retrieve data points by {\em $x$-coordinate}, a second copy of the list of data points, sorted according to their $x$-coordinate and augmented with pointers to the depth-sorted list, is maintained. Every point in the depth-sorted list is augmented with a pointer to its bucket, which allows a constant time retrieval of the depth of any point and a total running time that is linear in the set size.

\subsection{Algorithm}
Assume point $q$ is inserted into or deleted from the data set.
\begin{itemize}
\item Add or remove, if needed, a bucket to account for the new size of the data set.
\item For every point $p \in \mathcal{S}$
    \begin{itemize}
    \item Update the data structure associated with $p$ by {\em incrementing} or {\em decrementing} the number of points in each of the half-planes containing $q$
    \item If $q$ was inserted into the data set, update $p$'s data structure by \texttt{inserting} the two vectors associated with the meaningful half-planes determined by $p$ and $q$, setting the corresponding \texttt{to} fields according to the direction to $q$, and labeling their \texttt{Depth} as the depth associated with their \texttt{direction}.  The \texttt{Depth} can be computed by examining the neighboring half-planes' \texttt{Depth} and \texttt{to} fields.
    \item If $q$ was deleted from the data set, update $p$'s data structure by \texttt{Removing} the vectors associated with the meaningful half-planes determined by $q$.
    \item Update the depth of $p$ by considering the value of \texttt{minDepth} at the root of the tree.
    \item If the depth of $p$ changed, remove it from its bucket and place it in the adjacent bucket, according to its new \texttt{Depth}.
    \end{itemize}
\item If the data structure which was chosen for this algorithm is the one from Section 4 then, since this data structure also maintains the defining half-planes, the defining half-planes must be updated:
    \begin{itemize}
    \item Compare $q$ to the defining edges of $p$, and decide which of the cases described in Section \ref{sec:cases} applies. If needed, locate the new defining edges as described in the cases in Section \ref{sec:cases}.
    \item If $p$ is the unique point with maximum depth (maximum when compared to the other points), then use \texttt{Next} and \texttt{Prev} to determine if $p$ is on a degenerate contour as described in Section \ref{sec:cases}.  If $p$ is on a degenerate half-plane, update its defining half-planes to be \textsc{null}.
    \item If $p$ was on a degenerate contour, but is no longer on a degenerate contour, then use its data structure along with the converse counting property, Property \ref{prop:reverse-counting}, to determine its defining half-planes.
    \end{itemize}
\item If $q$ was inserted into the data set, build its data structure by sorting all vectors around $q$, assigning them the correct \texttt{Depths}, computing the half-space depth of $q$, and placing it in the correct bucket.
\item If $q$ was deleted from the data set, remove it from the correct bucket.
\end{itemize}

This data structure culminates in algorithms for maintaining the half-space depth of all points and the rank-based contours.

\begin{theorem}\label{thm:pointsort}
An algorithm exists that can compute the half-space depth ranking-order of a set of points dynamically in $O(n \log n)$ time per operation and quadratic space.  In addition, any rank-based contour can be reported in $O(n)$ time.
\end{theorem}
\begin{proof}
Note that we do not actually maintain the contours themselves, but a structure to report a contour efficiently.  Theorem \ref{thm:pointsort} implies that we can maintain the ranking-order of all points.  Since we also have kept the data points in order sorted by $x$-coordinate and the two structures linked, in linear time, we examine each bucket to select the appropriate number of deep data points and then run Graham-scan on these data points.  In this case, since the data points are already ordered by $x$-coordinate, the Graham-scan for the upper and lower convex hulls takes linear time, see \citep{Book:Berg:CompGeo}.
\end{proof}

\section{Dynamically Updating Cover Contours}\label{sec:DynamicCover}
We present an algorithm for computing the cover-based half-space depth contours of a set of points dynamically in $O(n\log^2n)$ time and overall quadratic space.
The algorithm is independent of the half-space depth and the rank-based contour algorithms.  This algorithm is based on the intersection of half-planes as described in Corollary \ref{cor:inter} and the relationship between half-plane intersections and convex hulls, see \citep[p.245]{Book:Berg:CompGeo}.  Our method is based on the dynamic computation of a convex hull, as described by Overmars and van Leeuwen \citep{Overmars:DynamicHull}.  Throughout this section, we {\em assume that no two data points have the same $x$-coordinate}.  If two points have the same $x$-coordinate, then either the entire data set can be rotated so that no two points have the same $x$-coordinate or, because the data set is in general position, a small (or infinitesimal) perturbation could be applied to one of these points.

\subsection{The Dual Arrangement}\label{sec:dualarrangement}
The algorithm is set in the dual space, see \citep{Book:Berg:CompGeo}: every point $p=(p_x,p_y)$ is mapped to the line $p^\ast$ with equation $y=p_xx-p_y$ and every line $l$ given by the equation $y=mx+b$ is mapped to the point $l^\ast$ with coordinates $(m,-b)$, see Figure \ref{fig:levelsone}.  The dual transformation operation is called $T$, and the dual transform has the following property: it preserves ordering, so if the point $p$ is above the line $l$, then the line $p^\ast$ is below the point $l^\ast$ by the same vertical distance.  This transformation is useful because the half-plane intersection problem is dual to the convex hull problem.  Corollary \ref{cor:inter} shows that in order to maintain the cover-based depth contour of depth $k/n$, we should intersect all half-planes containing $n-k+1$ data points and whose bounding line passes through two points.  We begin by dividing this intersection in two, by considering the subset of these half-planes that open upwards and the subset of these half-planes that open downwards.  In this discussion, we only consider the set of half-planes that open downwards because the set that open upwards can be treated similarly.  We know that the intersection of these half-planes corresponds to an upper convex hull problem in the dual.  We describe this hull next.

In the primal, each of the half-planes of interest is a half-plane that is above $n-k+1$ points and whose bounding line passes through two data points.  For such a half-plane, the bounding line is strictly above $n-k-1$ points.  A line that passes through two data points corresponds, in the dual, to the intersection point $h$ of two lines dual to data points.  Since the bounding line was strictly above $n-k-1$ data points, $h$ is a point strictly below $n-k-1$ lines dual to data points, see Figure \ref{fig:levelsone}.  Also, any intersection point of two dual lines that is strictly below $n-k-1$ dual lines corresponds to a half-plane containing $n-k+1$ data points, whose bounding line passes through two data points, and which opens downward.  Therefore, what is needed to compute the intersection of half-planes in the primal is the upper convex hull of these intersection points in the dual.

\begin{figure}[hbt]
\begin{center}
\begin{tabular}{m{1in}m{.2in}m{3in}}
\epsfysize=1.7in
\epsffile{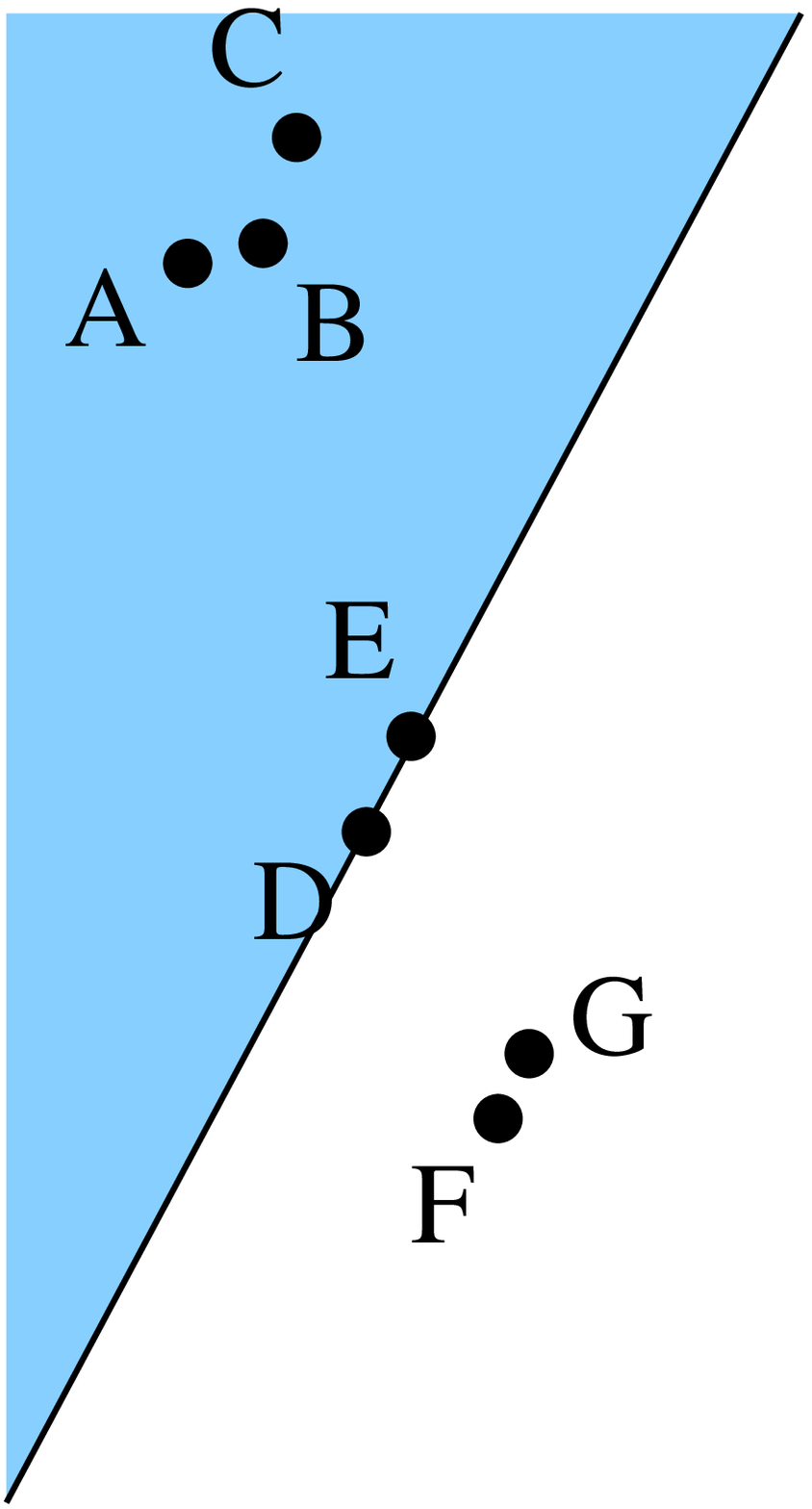}&
$\Leftrightarrow$&
\epsfysize=1.7in
\epsffile{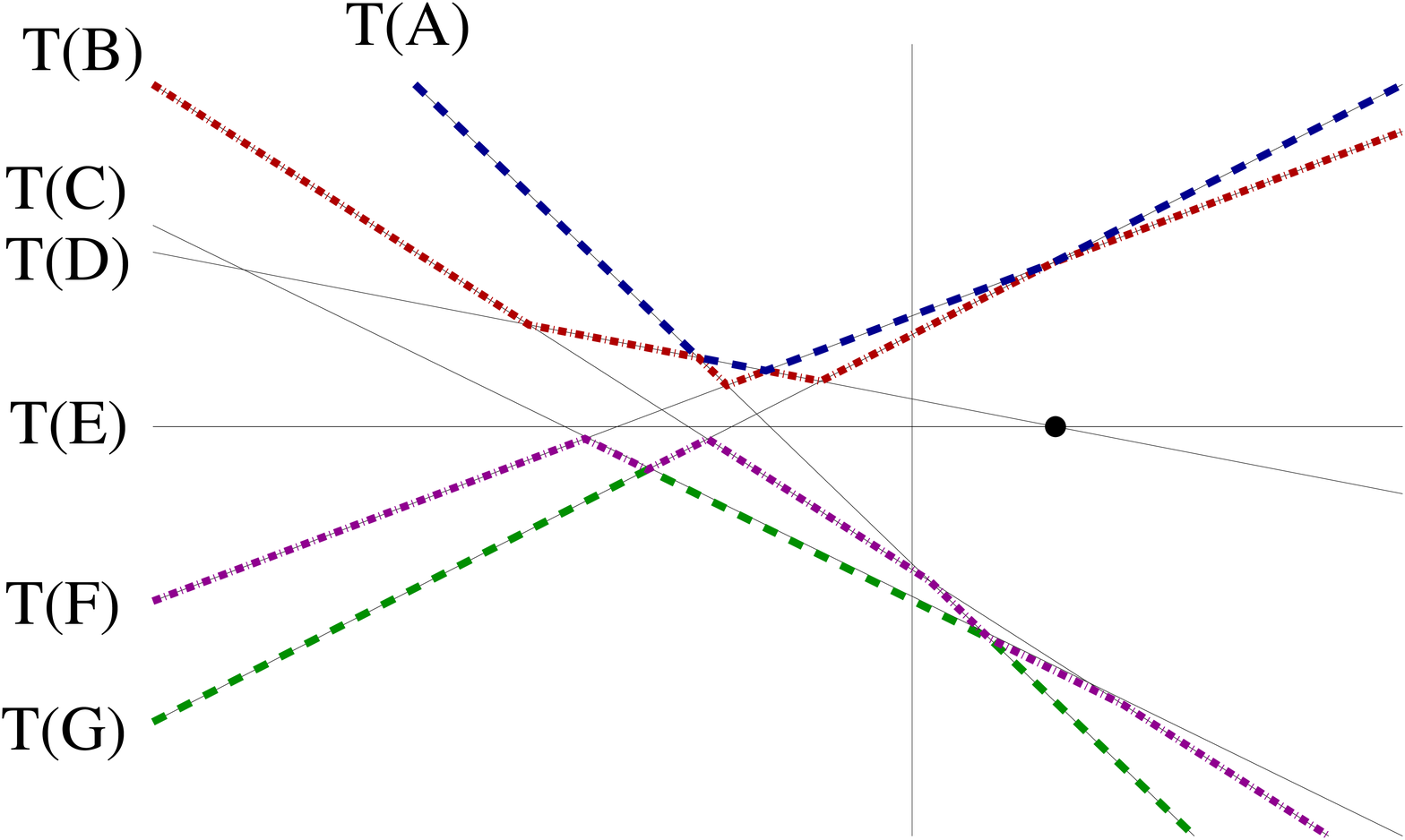}\\
\begin{center}(a)\end{center}&&\begin{center}(b)\end{center}
\end{tabular}
\end{center}
\caption{(a) A set of 7 points in the primal plane and (b) the corresponding set of lines in the dual plane under the dual transformation $T$ where the point $p=(a,b)$ is taken to the line $y=ax-b$.  In the dual plane, the dashed lines trace out the first and seventh levels (from top to bottom, respectively) while the dashed-dotted lines trace out the second and sixth levels (from top to bottom, respectively).  The highlighted line in the primal, which passes through points $D$ and $E$, defines two half-planes.  This line has three points strictly above it and two points strictly below it.  In addition, this line corresponds to dot at the intersection of $T(D)$ and $T(E)$ in the dual.  This dot has the two lines corresponding to $F$ and $G$ above it while it has the three lines corresponding to $A$, $B$, and $C$ below it.}
\label{fig:levelsone}
\end{figure}

In order to keep track of all of these intersection points with $n-k-1$ lines strictly above them, we look to the $(n-k+1)$th level of the dual arrangement (the levels are numbered starting with 1).  The $(n-k+1)$th level is the closure of the set of all segments that have $n-k$ lines above them, see Figure \ref{fig:levelsone}.  The level consists of segments and intersection points.  For any point on a segment, there are exactly $n-k$ lines above it.  For any intersection point, there are either $n-k$ or $n-k-1$ lines strictly above it.  Moreover, every intersection point with $n-k-1$ lines strictly above it appears on level $n-k+1$.  Finally, the upper hull of the intersection points on level $n-k+1$ is almost the upper hull of all intersection points with $n-k-1$ lines strictly above them.  The only caveat is that the endpoints, i.e., the first and last points on the hull could have $n-k$ lines above them.  This can be checked easily by considering the slopes of the two intersecting lines.  If the slopes of the segments on the level increases, then there are $n-k-1$ lines strictly above the intersection point and if the slope decreases, then there are $n-k$ lines strictly above the intersection point.  Therefore, our algorithm will maintain the upper hulls (and similarly, the lower hulls) of each level, which will determine the cover-based contours.

\subsection{Basic Data Structure}
The algorithm uses a data structure for updating the convex hull of a set of points dynamically in $O(\log^2 n)$ per update and in linear overall space from \citep{Overmars:DynamicHull,DobkinSouvaine:CompGeo}.  In particular, we maintain each level of the arrangement of lines in the dual in such a data structure.  The maintenance of the convex hull provides the contours as described above in Section \ref{sec:dualarrangement}.  In this section, we will call a connected component of an arrangement where the $x$-coordinates of the vertices are increasing a {\em chain}; a level of an arrangement is an example of a chain.  Our operations, given a data set, are as follows:

\begin{itemize}
\item $\texttt{Split(T,}p\texttt{)}$ - Given a chain $\texttt{T}$ from an arrangement and a data point $p$, compute the chains associated with the points to the left and right of $p$.  This operation was part of the original data structure of Overmars {\em et al.}
\item $\texttt{Join(T}_{\texttt{1}},\texttt{T}_{\texttt{2}}\texttt{,}p\texttt{)}$ - Given two chains $\texttt{T}_\texttt{1}$ and $\texttt{T}_\texttt{2}$ such that for every point $p_1\in\texttt{T}_\texttt{1}$ and $p_2\in\texttt{T}_\texttt{2}$, $xcoord(p_1)<xcoord(p)<xcoord(p_2)$.  Compute the chain that connects the two chains $\texttt{T}_\texttt{1}$ and $\texttt{T}_\texttt{2}$ via $p$.  This operation was not needed for the dynamic update of a single convex hull and cannot be used for joining general hulls, but, in this case, the ordering of vertices of the chains allows for this operation.
\end{itemize}

\subsubsection{Detailed Operations}
Our implementation uses a red-black tree augmented with the operations of a concatenable queue to enable additional {\em dynamic operations} to the data structure not required by the dynamic convex hull original algorithm:
\begin{itemize}
\item A $\texttt{Split(T,}p\texttt{)}$ is performed in $O(\log^2n)$ time as in Overmars {\em et al.}:  Start from the root and descend down the search path towards node $p$, in each node the partial hull is disassembled and reassembled at its 2 children.  After the hull has been split, the tree itself can be split and each portion can be rebalanced by local changes.
\item A $\texttt{Join(T}_{\texttt{1}},\texttt{T}_{\texttt{2}}\texttt{,}p\texttt{)}$ is performed in $O(\log^n)$ time.  If the black-node height of tree $\texttt{T}_\texttt{2}$ is less than the black-node height of the tree $\texttt{T}_\texttt{1}$.  Then $\texttt{T}_\texttt{2}$ is added as a subtree of $\texttt{T}_\texttt{1}$.  After the trees have been joined, the partial hulls along the path of the join are then rebuilt.
\end{itemize}
\subsection{Algorithm}
Our algorithm computes the contours by computing the upper and lower convex hulls of the intersection points along the levels. Each intersection point is defined using the two lines of the arrangement that create the intersection.
Two dynamic trees maintain the upper and lower convex hulls for every level of the arrangement, $2n$ trees for all levels and total of $O(n^2)$ space, the size of the entire arrangement.

When a new point $q$ is inserted into or deleted from the data set, some of the dual levels, the ones intersected by the line $T(q)$, and consequently some of the convex hulls of these levels, are affected. We compute the new convex hulls by updating the convex hull trees.
\subsubsection{Insertions}
Given a data point $q$ (and its dual $T(q)$), a set of $n$ points $\mathcal{F}_n$, and its associated data structures, we wish to recompute the half-space depth contours of $\mathcal{F}_n \cup \{q\}$.
After each insertion, each {\it old} level $k$, intersected by $T(q)$, splits between levels $k$ and $k+1$.  Therefore, the new level $k$ becomes a combination of chains from old level $k$ that remain in level $k$ plus chains from old level $k-1$ that increase level to level $k$ minus chains from old level $k$ that increase level to level $k+1$.  In every stitching point between these chains, a segment of $T(q)$ is added.  Not all levels are necessarily updated, only those intersected by $T(q)$, i.e., the old levels between $k'+1$ and $n-k'$ are affected, where the half-space depth of $q$ is $k'/n$ before the insertion.
Since $T(q)$ intersects the arrangement exactly $n$ times, $O(n)$ chains move between all trees and $n$ segments from $T(q)$ are added.

To update the contours properly, we need to stitch the chains correctly.
The stitching is done by splitting the dynamic trees representing the chains and merging the subtrees, while maintaining a correct representation of the convex hull on the root of each tree.

\begin{figure}[hbt]
\begin{center}
\begin{tabular}{cc}
\epsfysize=1.4in
\epsffile{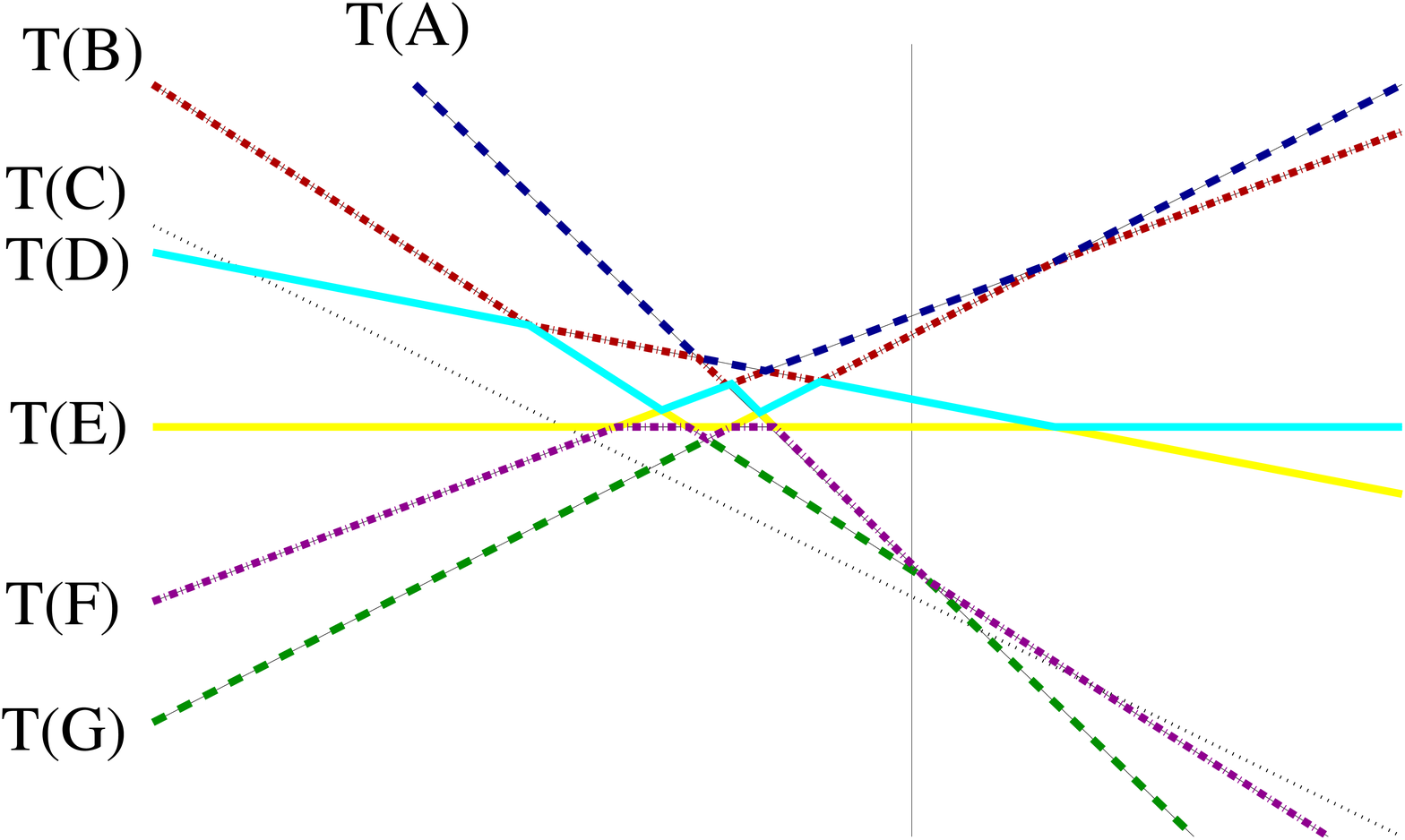}&
\epsfysize=1.4in
\epsffile{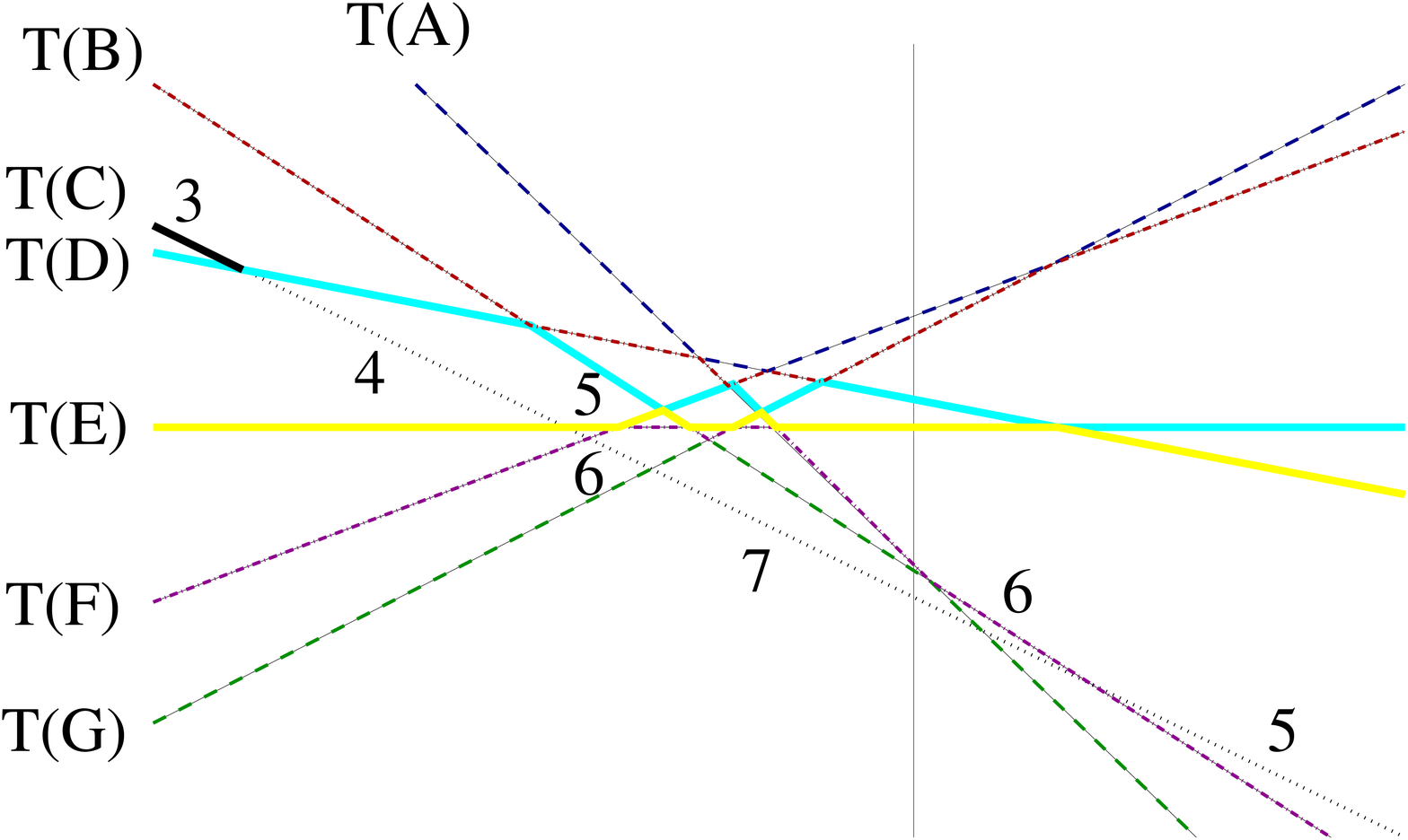}\\
(a)&(b)\\
\epsfysize=1.4in
\epsffile{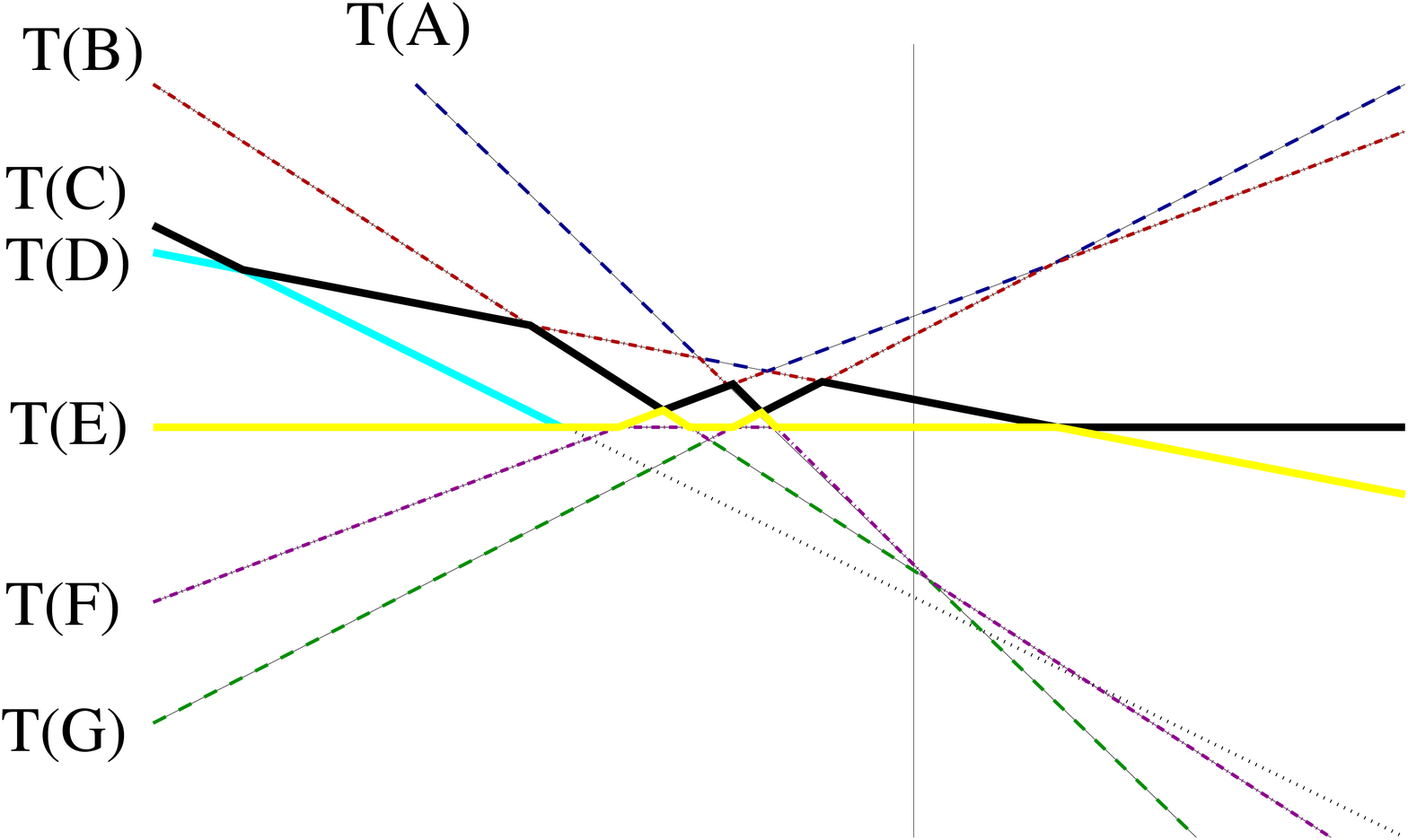}&
\epsfysize=1.4in
\epsffile{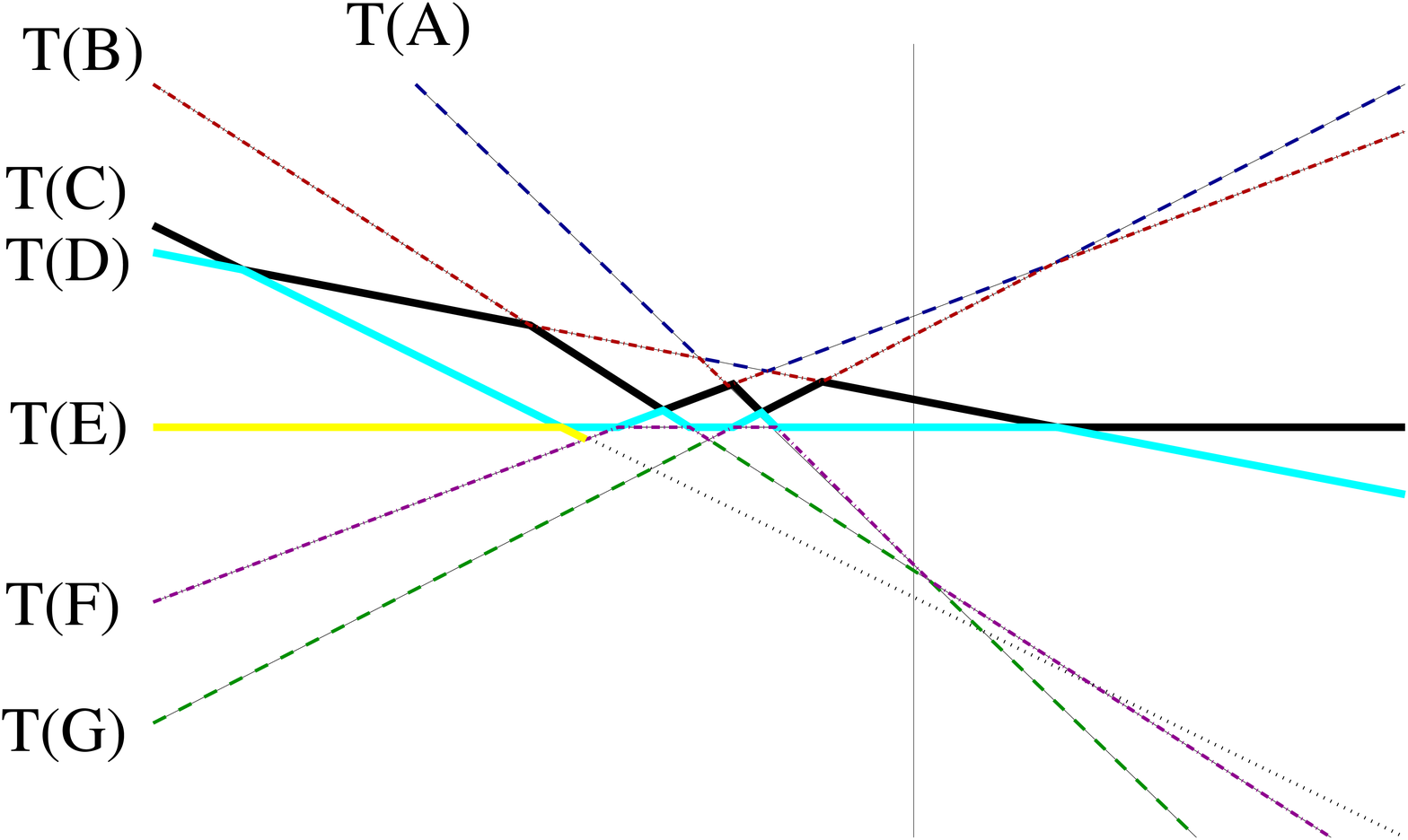}\\
(c)&
(d)\\
\end{tabular}
\end{center}
\caption{(a) The six levels of the arrangement before the arrangement is augmented with $T(C)$.  (b) $T(C)$ has been subdivided where it intersects the other lines of the arrangement.  Each interval (or ray) has been labeled according to its level in the new arrangement.  The new level 3 has been started at $T(C)$.  (c) When the new level 3 intersects the old level 3, the chains switch so the right part of the old level 3 is stitched to the new level 3 and the left part of the old level 3 becomes new level 4.  In addition, the next segment along $T(C)$ is stitched to the new level 4.  (d) When the new level 4 intersects the old level 4, the chains switch in a similar manner as to the stiching in (c).  The next step will be to switch the new level 5 and the old level 5.}
\label{fig:levelstwo}
\end{figure}

\subsubsection{Detailed Operations}
\begin{itemize}
\item \label{step1} Compute the intersection point of $T(q)$ with every line $l$ of the arrangement and sort them according to $x$ coordinate. Since there are only $n$ other lines (and $n$ intersections), this step takes $O(n\log n)$.
\item Start a new chain with the leftmost segment of $T(q)$, at the level corresponding to the order statistic of the slope $T(q)$ with respect to the arrangement.
\item Increment all levels above this new level by 1.
\item Find all of the segments of $T(q)$, and find the arrangement each will join when $q$ is added to the data set.  This can be done in $O(n)$ time since the difference in the level between two sequential segments is $\pm1$.  The choice of $\pm$ is determined comparing the slopes of the two lines at the intersection.
\item Construct the new set of levels by splitting and joining trees, and inserting segments of $T(q)$: Starting with the new level consisting of the leftmost segment of $T(q)$, for each segment of $T(q)$, in order, \texttt{Split} the appropriate chain at the intersection point (the right endpoint of the segment).  Next, \texttt{Join} the top two chains and the bottom two chains at the split point.  Finally, \texttt{Join} the remaining chain to the next segment of $T(q)$.  Continue this until there are no more segments of $T(q)$ to attach.  This step takes $O(\log^2n)$ for each \texttt{Split} and \texttt{Join}, and there are a total of $O(n)$ of these steps (there are a constant number of steps for each segment of $T(q)$).  This results in a total time complexity for the algorithm of $O(n\log^2n)$.
\end{itemize}

\subsubsection{Deletions}
Given a set of points $\mathcal{F}_n$, its associated data structures, and a data point $q$ (and its dual line $T(q)$) we wish to recompute the structure associated with $\mathcal{F}_n \setminus\{q\}$.  When a line $T(q)$ is removed from the arrangement, level $k$ splits between levels $k$ and $k-1$.  Similar to the insertion case, the stitching of the levels must be done appropriately, while maintaining a correct representation of the convex hull at the root of each tree.

\subsubsection{Detailed Operations}
\begin{itemize}
\item \label{step2} Compute the intersection point of $T(q)$ with every line $l$ of the arrangement and sort them according to $x$ coordinate. Since there are only $n$ other lines (and $n$ intersections), this step takes $O(n\log n)$.
\item Find all of the segments of $T(q)$ in the arrangement, and find the level each will leave when $q$ is removed from the data set.  This can be done in $O(n)$ time since the difference in the level between two sequential segments is $\pm1$.  The choice of $\pm$ is determined by comparing the slopes of the lines at the intersection.- $O(n)$.
\item Construct the new set of levels by splitting and merging trees, and removing segments of $T(q)$: Starting with the level consisting of the leftmost segment of $T(q)$, for each segment of $T(q)$, in order, $\texttt{Split}$ both chains at the intersection point (the right endpoint of the segment).  Next, \texttt{Join} the two chains which do not involve segments of $T(q)$.  Discard the chains consisting only of segments of $T(q)$.  Continue this until there are no more segments of $T(q)$ to discard.  This step takes $O(\log^2n)$ for each $\texttt{Split}$ and $\texttt{Join}$, and there are a total of $O(n)$ of these steps (there are a constant number of steps for each segment of $T(q)$).  This results in a total time complexity for the algorithm of $O(n\log^2n)$.
\item Decrement all levels which started below the leftmost segment of $T(q)$.
\end{itemize}

\begin{theorem}
It is possible to compute the cover-based half-space depth contours dynamically in $O(n\log^2 n)$ time per update and $O(n^2)$ overall space.  Moreover, any contour can be reported in $O(m)$ time, where $m$ is complexity of the corresponding levels in the arrangement.
\end{theorem}
\begin{proof}
Note that we do not actually maintain the contours themselves, but a structure to report a contour efficiently.  In order to compute the cover-based contour of depth $k/n$, take the upper hull of the $(n-k+1)$th level and the and the lower hull of the $k$th level, removing the endpoints if they are not appropriate, and returning these to the primal, where they represent the intersection of half-planes containing $n-k+1$ data points and whose bounding line passes through two data points.  Finally, the intersection of two convex regions needs to be performed, which can be done in the complexity of the hulls because the hulls are naturally sorted by slope, which corresponds to being sorted by $x$-coordinate in the primal plane.  The convexity of the hulls is bounded by the convexity of the levels themselves.
\end{proof}

\section{Conclusion}
In this paper, we provided data structures and algorithms for dynamically computing half-space depth of points as well as rank-based and contour-based contours.  The depth of a single point can be maintained in $O(\log n)$ time per update using overall linear space.  The rank-based contours can be maintained in $O(n\log n)$ time per update using overall quadratic space.  The cover-based contours can be maintained in $O(n\log^2 n)$ time per update using overall quadratic space.

In addition, we studied the local cover-based contours near data points and described how they change as data points are added or deleted from the data set.  In this direction, we proved a result of independent interest that at least one local cover-based contour edge does at a point does not change as points are inserted into or deleted from the data set.

We end with the open problem of finding additional dynamic algorithms for other data depth functions as desired by the probability and statistics community, \citep{Liu:Personal}.  One would expect that this problem may be significantly harder than its half-space depth counterpart because most other data-depth functions are not as well-behaved, in the finite sample cases, has half-space depth, see \citep{Zuo:Structural,SimplicialDepth}.

\bibliographystyle{plainnat}
\bibliography{DataDepth}

\begin{thebibliography}{41}
\providecommand{\natexlab}[1]{#1}
\providecommand{\url}[1]{\texttt{#1}}
\expandafter\ifx\csname urlstyle\endcsname\relax
  \providecommand{\doi}[1]{doi: #1}\else
  \providecommand{\doi}{doi: \begingroup \urlstyle{rm}\Url}\fi

\bibitem[Abbott et~al.(2009)Abbott, Burr, Chan, Demaine, Demaine, Hugg, Kane,
  Langerman, Nelson, Rafalin, Seyboth, and Yeunga]{DynamicHam}
Timothy~G. Abbott, Michael~A. Burr, Timothy~M. Chan, Erik~D. Demaine, Martin~L.
  Demaine, John Hugg, Daniel Kane, Stefan Langerman, Jelani Nelson, Eynat
  Rafalin, Kathryn Seyboth, and Vincent Yeunga.
\newblock Dynamic ham-sandwich cuts in the plane.
\newblock \emph{Computational Geometry}, 42\penalty0 (5):\penalty0 419--428,
  2009.

\bibitem[Agarwal and Matou{\v{s}}ek(1995)]{Agarwal:Dynamic}
Pankaj~K. Agarwal and Ji{\v{r}}{\'i} Matou{\v{s}}ek.
\newblock Dynamic half-space range reporting and its applications.
\newblock \emph{Algorithmica}, 13\penalty0 (4):\penalty0 325--345, 1995.

\bibitem[Agarwal et~al.(2008)Agarwal, Sharir, and Welzl]{Agarwal:Center}
Pankaj~K. Agarwal, Micha Sharir, and Emo Welzl.
\newblock Algorithms for center and tverberg points.
\newblock \emph{ACM Transactions on Algorithms}, 5\penalty0 (1), 2008.

\bibitem[Aloupis(2006)]{Aloupis:DIMACS:DataDepth}
Greg Aloupis.
\newblock Geometric measures of data depth.
\newblock In  \citet{Book:DIMACS:DataDepth}, pages 147--158.

\bibitem[Aloupis et~al.(2002)Aloupis, Cort{\'e}s, G{\'o}mes, Soss, and
  Toussaint]{Aloupis:LowerBounds}
Greg Aloupis, Carmen Cort{\'e}s, Francisco G{\'o}mes, Michael Soss, and
  Godfried Toussaint.
\newblock Lower bounds for computing statistical depth.
\newblock \emph{Computational Statistics {\&} Data Analysis}, 40\penalty0
  (2):\penalty0 223--229, 2002.

\bibitem[Barnett(1976)]{Barnett:Peeling}
Vic Barnett.
\newblock The ordering of multivariate data.
\newblock \emph{Journal of the Royal Statistical Society. Series A (General)},
  139\penalty0 (3):\penalty0 318--355, 1976.

\bibitem[Bremner et~al.(2008)Bremner, Chen, Iacono, Langerman, and
  Morin]{Bremner:half-spaceDepth}
David Bremner, Dan Chen, John Iacono, Stefan Langerman, and Pat Morin.
\newblock Output-sensitive algorithms for tukey depth and related problems.
\newblock \emph{Statistics and Computing}, 18\penalty0 (3):\penalty0 259--266,
  2008.

\bibitem[Burr et~al.(2006)Burr, Rafalin, and Souvaine]{SimplicialDepth}
Michael~A. Burr, Eynat Rafalin, and Diane~L. Souvaine.
\newblock Simplicial depth: An improved definition, analysis, and efficiency
  for the finite sample case.
\newblock In  \citet{Book:DIMACS:DataDepth}, pages 195--209.

\bibitem[Chan(2004)]{Chan:Randomized}
Timothy~M. Chan.
\newblock An optimal randomized algorithm for maximum tukey depth.
\newblock \emph{Symposium on Discrete Algorithms}, pages 430--436, 2004.

\bibitem[Cole(1987)]{Cole:Centerpoint2}
Richard Cole.
\newblock Slowing down sorting networks to obtain faster sorting algorithms.
\newblock \emph{Journal of the ACM}, 34\penalty0 (1):\penalty0 200--208, 1987.

\bibitem[Cole et~al.(1987)Cole, Sharir, and Yap]{Cole:Hulls}
Richard Cole, Micha Sharir, and Chee~K. Yap.
\newblock On k-hulls and related problems.
\newblock \emph{SIAM Journal on Computing}, 16\penalty0 (1):\penalty0 61--77,
  1987.

\bibitem[de~Berg et~al.(2000)de~Berg, van Krefeld, Overmars, and
  Schwarzkopf]{Book:Berg:CompGeo}
Mark de~Berg, M.~van Krefeld, M.~Overmars, and O.~Schwarzkopf.
\newblock \emph{Computational Geometry: Algorithms and Applications}.
\newblock Springer, 2nd edition, 2000.

\bibitem[Dobkin and Souvaine(1987)]{DobkinSouvaine:CompGeo}
David Dobkin and Diane Souvaine.
\newblock Computational geometry -- a user's guide.
\newblock In Jacob~T. Schwartz and Chee-Keng Yap, editors, \emph{Algorithmic
  and geometric aspects of robotics}, volume~1 of \emph{Advances in Robotics},
  chapter~2. Lawrence Erlbaum Associates Inc., Hillsdale, NJ, 1987.

\bibitem[Eckhoff(1993)]{Eckhoff:Helly}
J{\"u}rgen Eckhoff.
\newblock Helly, radon, and carath{\'e}odory type theorems.
\newblock In  \citet{Book:Handbook:Convex}, pages 389--448.

\bibitem[Eddy(1982)]{Eddy:Peeling}
W.~F. Eddy.
\newblock Convex hull peeling.
\newblock In \emph{COMPSTAT 1982: Proceedings in COmputational Statistics, Part
  1}, pages 42--47. Physica-Verlag, 1982.

\bibitem[Edelsbrunner and Welzl(1986)]{EdelsbrunnerWelzelBelts}
Herbert Edelsbrunner and Emo Welzl.
\newblock Constructing belts in two-dimensional arrangements with applications.
\newblock \emph{SIAM Journal on Computing}, 15\penalty0 (1):\penalty0 271--284,
  1986.

\bibitem[Floudas and Pardalos(2003)]{Book:Optimization}
Christodoulos~A. Floudas and Panos~M. Pardalos, editors.
\newblock \emph{Frontiers in Global Optimization}, volume~74 of \emph{Nonconvex
  Optimization and Its Applications}.
\newblock Springer, 2003.

\bibitem[Fukuda and Rosta(2003)]{Fukuda:Parallel}
Komei Fukuda and Vera Rosta.
\newblock Exact parallel algorithms for the location depth and the maximum
  feasible subsystem problems.
\newblock In  \citet{Book:Optimization}, pages 123--133.

\bibitem[Gruber and Wills(1993)]{Book:Handbook:Convex}
Peter~M. Gruber and J{\"o}rg~M. Wills, editors.
\newblock \emph{Handbook of Convex Geometry}.
\newblock North-Holland, Amsterdam, 1993.

\bibitem[Hodges(1955)]{Hodges:Bivariate}
Joseph~L. Hodges.
\newblock A bivariate sign test.
\newblock \emph{Annals of Mathematical Statistics}, 26\penalty0 (3):\penalty0
  523--527, 1955.

\bibitem[Jadhav and Mukhopadhyay(1994)]{Jadhav:Centerpoint}
Shreesh Jadhav and Asish Mukhopadhyay.
\newblock Computing a centerpoint of a finite planar set of points in linear
  time.
\newblock \emph{Discrete {\&} Computational Geometry}, 12\penalty0
  (3):\penalty0 291--312, 1994.

\bibitem[Johnson et~al.(1998)Johnson, Kwok, and Ng]{Johnson:Fast}
Ted Johnson, Ivy Kwok, and Raymond Ng.
\newblock Fast computation of two-dimensional depth contours.
\newblock \emph{The Fourth International Conference on Knowledge Discovery and
  Data Mining}, pages 224--228, 1998.

\bibitem[Krishnan et~al.(2002)Krishnan, Mustafa, and
  Venkatasubramanian]{Krishnan:Hardware}
Shankar Krishnan, Nabil~H. Mustafa, and Suresh Venkatasubramanian.
\newblock Hardware-assisted computation of depth contours.
\newblock \emph{Symposium on Discrete Algorithms}, pages 558--567, 2002.

\bibitem[Langerman and Steiger(2000)]{Langerman:Maximal}
Stefan Langerman and William~L. Steiger.
\newblock Computing a maximal depth point in the plane.
\newblock \emph{Proceedings of the Japan Conference on Discrete and
  Computational Geometry}, 2000.

\bibitem[Langerman and Steiger(2003)]{Langerman:Median}
Stefan Langerman and William~L. Steiger.
\newblock Optimization in arrangements.
\newblock In \emph{Proceedings of the 20th Annual Symposium on Theoretical
  Aspects of Computer Science}, volume 2607 of \emph{Lecture Notes In Computer
  Science}, pages 50--61, 2003.

\bibitem[Liu(2003)]{Liu:Personal}
Regina Liu.
\newblock Personal communication, May 2003.

\bibitem[Liu et~al.(2006)Liu, Serfling, and Souvaine]{Book:DIMACS:DataDepth}
Regina Liu, Robert Serfling, and Diane Souvaine, editors.
\newblock \emph{Data Depth: Robust Multivariate Analysis, Computational
  Geometry and Applications}, volume~72 of \emph{DIMACS Series in Discrete
  Mathematics and Theoretical Computer Science}.
\newblock American Mathematical Society, 2006.

\bibitem[Liu et~al.(1999)Liu, Parelius, and Singh]{Liu:Survey}
Regina~Y. Liu, Jesse~M. Parelius, and Kesar Singh.
\newblock Multivariate analysis by data depth: descriptive statistics, graphics
  and inference.
\newblock \emph{Annals of Statistics}, 27\penalty0 (3):\penalty0 783--858,
  1999.

\bibitem[Matou{\v{s}}ek(1991)]{Matousek:Median}
Ji{\v{r}}{\'i} Matou{\v{s}}ek.
\newblock Computing the center of planar point sets.
\newblock In \emph{Discrete and Computational Geometry: Papers from the DIMACS
  Special Year}, volume~6 of \emph{DIMACS Series in Discrete Mathematics and
  Theoretical Computer Science}, pages 221--230. American Mathematical Society,
  1991.

\bibitem[Miller et~al.(2003)Miller, Ramaswami, Rousseeuw, Sellar{\`e}s,
  Souvaine, Streinu, and Struyf]{Miller:Contours}
Kim Miller, Suneeta Ramaswami, Peter Rousseeuw, J.~Antoni Sellar{\`e}s, Diane
  Souvaine, Ileana Streinu, and Anja Struyf.
\newblock Efficient computation of location depth contours by methods of
  computational geometry.
\newblock \emph{Journal Statistics and Computing}, 13\penalty0 (2):\penalty0
  153--162, April 2003.

\bibitem[Overmars and van Leeuwen(1981)]{Overmars:DynamicHull}
Mark~H. Overmars and Jan van Leeuwen.
\newblock Maintenance of configurations in the plane.
\newblock \emph{Journal of Computer and System Sciences}, 23\penalty0
  (2):\penalty0 166--204, 1981.

\bibitem[Rafalin and Souvaine(2004)]{Rafalin:Contours}
Eynat Rafalin and Diane~L. Souvaine.
\newblock Data depth contours - a computational geometry perspective.
\newblock Technical Report~2, Tufts University, Computer Science Department,
  May 2004.

\bibitem[Rousseeuw and Struyf(1998)]{Rousseeuw:HighDimensions}
Peter~J. Rousseeuw and Anja Struyf.
\newblock Computing location depth and regression depth in higher dimensions.
\newblock \emph{Statistics and Computing}, 8\penalty0 (3):\penalty0 193--203,
  1998.

\bibitem[Rousseeuw et~al.(1999)Rousseeuw, Ruts, and Tukey]{Rousseeuw:Bagplot}
Peter~J. Rousseeuw, Ida Ruts, and John~W. Tukey.
\newblock The bagplot: a bivariate boxplot.
\newblock \emph{The American Statistician}, 53\penalty0 (4):\penalty0 382--387,
  1999.

\bibitem[Ruts and Rousseeuw(1996{\natexlab{a}})]{Ruts:Bivariate}
Ida Ruts and Peter Rousseeuw.
\newblock Bivariate location depth.
\newblock \emph{Applied statistics}, 45\penalty0 (4):\penalty0 516--526,
  1996{\natexlab{a}}.

\bibitem[Ruts and Rousseeuw(1998)]{Ruts:Bivariate2}
Ida Ruts and Peter Rousseeuw.
\newblock Constructing the bivariate tukey median.
\newblock \emph{Statistica Sinica}, 8\penalty0 (3):\penalty0 827--839, 1998.

\bibitem[Ruts and Rousseeuw(1996{\natexlab{b}})]{Ruts:Contours}
Ida Ruts and Peter Rousseeuw.
\newblock Computing depth contours of bivariate point clouds.
\newblock \emph{Computational Statistics {\&} Data Analysis}, 23\penalty0
  (1):\penalty0 153--168, 1996{\natexlab{b}}.

\bibitem[Struyf and Rousseeuw(2000)]{Struyf:HighDimensions}
Anja Struyf and Peter~J. Rousseeuw.
\newblock High-dimensional computation of the deepest location.
\newblock \emph{Computational Statistics {\&} Data Analysis}, 34\penalty0
  (4):\penalty0 415--426, 2000.

\bibitem[Tukey(1975)]{Tukey:half-space}
John~W. Tukey.
\newblock Mathematics and the picturing of data.
\newblock In \emph{Proceedings of the International Congress of Mathematicians
  (Vancouver, B. C., 1974)}, volume~2, pages 523--531, Montreal, Quebec, 1975.
  Canadian Mathematical Congress.

\bibitem[van Kreveld et~al.(1999)van Kreveld, Mitchell, Rousseeuw, Sharir,
  Snoeyink, and Speckmann]{Kreveld:Regression}
Marc van Kreveld, Joseph S.~B. Mitchell, Peter Rousseeuw, Micha Sharir, Jack
  Snoeyink, and Bettina Speckmann.
\newblock Efficient algorithms for maximum regression depth.
\newblock \emph{Symposium on Computational Geometry}, pages 31--40, 1999.

\bibitem[Zuo and Serfling(2000)]{Zuo:Structural}
Yijun Zuo and Robert Serfling.
\newblock Structural properties and convergence results for contours of sample
  statistical depth functions.
\newblock \emph{The Annals of Statistics}, 28\penalty0 (2):\penalty0 483--499,
  2000.

\end{thebibliography}

\end{document}